\documentclass[twocolumn]{IEEEtran}

\ifCLASSINFOpdf
  \usepackage[pdftex]{graphicx}
  \usepackage{epstopdf}
  \graphicspath{{./}}
  \DeclareGraphicsExtensions{.pdf,.jpeg,.png}
\else
  \usepackage[dvips]{graphicx}
  \graphicspath{{./}}
  \DeclareGraphicsExtensions{.eps}
\fi

\usepackage[cmex10]{amsmath}
\usepackage{amssymb,amsthm,amsfonts}
\usepackage{graphicx}
\usepackage{cite}
\usepackage{mathabx}
\usepackage{accents} % to put a underline bar in the definition of the vect. Should be loaded after mathabx 
\usepackage{flushend}
\usepackage{xcolor}

\newtheorem{theorem}{Theorem}
\newtheorem{lemma}[theorem]{Lemma}
\newtheorem{corollary}{Corollary}

\theoremstyle{definition}
\newtheorem{definition}{Definition}
\theoremstyle{remark}
\newtheorem{remark}{Remark}
\theoremstyle{remark}
\newtheorem{assumption}{Assumption}

\DeclareMathOperator{\sinc}{sinc}
\DeclareMathOperator{\vol}{V}

\DeclareMathOperator{\diag}{diag}

\providecommand{\norm}[1]{\left\lVert#1\right\rVert}
\providecommand{\abs}[1]{\left|#1\right|}
\newcommand{\nn}{\nonumber}
\newcommand{\vect}[1]{\underaccent{\bar}{#1}}
\newcommand{\const}[1]{{\mathcal{#1}}}
\newcommand{\Reals}{\mathbb{R}}
\newcommand{\Complex}{\mathbb{C}}
\newcommand{\E}{\mathsf{E}}
\newcommand{\ie}{\emph{i.e.}}
\newcommand{\eg}{\emph{e.g.}}
\newcommand{\der}{\mathrm{d}}
\newcommand{\iid}{\text{i.i.d.}}
\newcommand{\eqdef}{\stackrel{\Delta}{=}}
\newcommand{\snr}{\text{SNR}}
\newcommand{\matd}[1]{\mathsf{#1}}
\newcommand{\matr}[1]{\mathbb{#1}}

\newcommand{\normalc}[2]{\mathcal{N}_{\Complex}\!\left(#1,#2\right)}

\newcommand{\pr}{\textnormal{Pr}}

\newcommand{\D}{\mathcal{D}}
\newcommand{\dof}{DOF}
\newcommand{\dofs}{DOFs}

\begin{document}

\thispagestyle{empty}

\title{The Asymptotic Capacity of the Optical Fiber%
\thanks{The author is with the Communications and Electronics
Department, T\'el\'ecom ParisTech, Paris, France.
Email: \texttt{yousefi@telecom-paristech.fr}.}
}

\author{Mansoor~I.~Yousefi}

\date{}

\maketitle

\IEEEpeerreviewmaketitle

\begin{abstract}
It is shown that signal energy is the only available degree-of-freedom (\dof)
for fiber-optic transmission as 
the input power tends to infinity. With $n$ signal \dofs\
at the input, $n-1$ \dofs\ are asymptotically lost to signal-noise interactions.
The main observation is that, nonlinearity 
introduces a multiplicative noise in the channel, similar to fading in wireless 
channels. The channel is viewed in the spherical coordinate system, where signal 
vector $\vect X\in\Complex^n$ is represented in terms of its norm $\abs{\vect X}$ and direction 
$\hat{\vect X}$. The multiplicative noise causes signal direction 
$\hat{\vect X}$ to vary randomly on the surface of the unit $(2n-1)$-sphere in $\Complex^{n}$, in such a way that
the effective area of the support of $\hat {\vect X}$ does not vanish as $\abs{\vect X}\rightarrow\infty$. On the other hand, 
the surface area of the sphere is finite, so that
$\hat{\vect X}$ carries finite information. This observation is 
used to show several results. Firstly, let $\const C(\const P)$ be the
capacity of a discrete-time periodic model of the optical fiber with
distributed noise and frequency-dependent loss, 
as a function of the average input power 
$\const P$. 
It is shown that asymptotically as $\const P\rightarrow\infty$,  $\const
C=\frac{1}{n}\log\bigl(\log\const P\bigr)+c$, where $n$ is the dimension of
the input signal space and $c$ is a bounded
number. 
In particular, 
$\lim_{\const P\rightarrow\infty}\const C(\const P)=\infty$ in
finite-dimensional periodic models.  
Secondly, it is shown that capacity 
saturates to a constant in infinite-dimensional models where $n=\infty$. An expression is provided for 
the constant $c$, by showing that, as the input $\abs{\vect
  X}\rightarrow\infty$, the action of the discrete periodic stochastic nonlinear
Schr\"odinger equation tends to multiplication by a random
matrix (with fixed distribution, independent of input). Thus,  perhaps
counter-intuitively, noise simplifies the nonlinear channel at high
powers to a 
\emph{linear} multiple-input multiple-output fading
channel. As $\const P\rightarrow\infty$ signal-noise interactions gradually reduce the
slope of the $\const C(\const P)$, to a point where increasing the
input power returns diminishing gains.
Nonlinear frequency-division 
multiplexing can be applied to approach capacity in optical
networks, where linear multiplexing achieves low rates at high powers.
 
\end{abstract}

%%%%%%%%
%SECTION I: Introduction
%%%%%%%

\section{Introduction}
Several decades since the introduction of the optical fiber, channel
capacity at high powers remains a vexing conundrum.  Existing achievable rates saturate at high powers because 
of linear multiplexing and treating the resulting interference as
noise in network environments \cite{yousefi2012nft1,yousefi2012nft2,yousefi2012nft3}. Furthermore, 
it is difficult to estimate the
capacity via numerical 
simulations, because channel has memory. 

Multi-user communication problem for (an ideal model of) optical fiber can be reduced to single-user 
problem using the nonlinear frequency-division multiplexing (NFDM) \cite{yousefi2012nft1,yousefi2012nft3}. This 
addresses deterministic distortions, such as inter-channel and inter-symbol 
interference (signal-signal interactions). The problem is then reduced to finding  
the capacity of the point-to-point optical fiber set by noise. 

There are two effects in fiber that impact Shannon capacity in point-to-point channels. (1) 
Phase noise. Nonlinearity transforms additive noise to phase noise in the channel. As the amplitude of the input signal 
tends to infinity, 
the phase of the output signal tends to a uniform random variable in
the zero-dispersion channel \cite[Section~IV]{yousefi2011opc}. 
As a result,  phase carries finite information in  the non-dispersive fiber.
(2) Multiplicative noise. Dispersion converts phase noise to amplitude noise, introducing an effect which at 
high powers is similar to fading in wireless channels. Importantly,
the conditional entropy grows strongly with input signal.

In this paper, we study the asymptotic capacity of a discrete-time
periodic model of the optical fiber as the 
input power tends to infinity. 
The role of the nonlinearity in point-to-point discrete channels pertains to signal-noise 
interactions, captured by the conditional entropy. 

The main result is the following theorem, describing capacity-cost function in models with constant and
non-constant loss; see Definition~\ref{def:loss}.

\begin{theorem} 
Consider the discrete-time periodic model of the NLS channel
\eqref{eq:nls} described in Section~\ref{sec:mssfm}, with non-zero dispersion. Capacity 
is asymptotically  
\begin{IEEEeqnarray*}{rCl}
\const C(\const P)=
\begin{cases}
\frac{1}{n}\log(\log\const P)+c, & \textnormal{  non-constant loss},\\
\frac{1}{2n}\log\const P+c, & \textnormal{constant loss},
\end{cases}
\end{IEEEeqnarray*}
where $n$ is dimension of the input signal space, $\const
P\rightarrow\infty$ is the average input signal power and $c\eqdef
c(n,\const P)<\infty$. In particular, 
$\lim\limits_{\const P\rightarrow\infty} \const C(\const P)=\infty$ in
finite-dimensional models. Intensity modulation and direct detection
(photon counting) is nearly
capacity-achieving in the limit $\const P\rightarrow\infty$, where
capacity is dominated by the first terms in $\const C(\const P)$ expressions.
\label{thm:main}
\end{theorem}

From the Theorem~\ref{thm:main} and \cite[Theorem 1]{yousefi2011opc}, the
asymptotic capacity of the dispersive fiber is much 
smaller than the asymptotic capacity of (the discrete-time model of)
the zero-dispersion 
fiber, which is $\frac{1}{2}\log\const P+c$, $c<\infty$. Dispersion
reduces the capacity, by increasing the conditional entropy. With $n$
\dofs\ at the input, $n-1$ \dofs\ are asymptotically lost to
signal-noise interactions, leaving signal energy as the only useful
\dof\ for transmission.

There are a finite number of \dofs\ in all computer
simulations and physical systems. 
However, as a mathematical problem, the following
Corollary holds true.

\begin{corollary} 
Capacity saturates to a constant $c<\infty$ in infinite-dimensional
models, including the continuous-time model.
\label{cor:inf}
\end{corollary}

The power level where signal-noise interactions begin to appreciably impact the
slope of the  $\const C(\const P)$ is not determined in this paper. Numerical simulations indicate that 
the conditional entropy does not increase with input in the nonlinear
Fourier domain, for a range of power larger than 
the optimal power in wavelength-division multiplexing \cite[Fig.~9 (a)]{yousefi2016nfdm}.   
In this regime, signal-noise interactions are weak and the capacity is
dominated by the (large) number $c$ in the Theorem~\ref{thm:main}. 
A numerical estimation of the capacity of the point-to-point fiber 
at input powers higher than those in Fig.~\ref{fig:nfdm} should reveal 
the impact of the signal-dependent noise on the asymptotic capacity.

The contributions of the paper are presented as follows. 
The continuous-time model is discretized
in Section~\ref{sec:mssfm}. The main ingredient is a modification
of the split-step Fourier method (SSFM) that shows noise influence more directly
compared with the standard SSFM. A \emph{unit} is defined in the modified SSFM (MSSFM) model that plays an
important role throughout the paper. The MSSFM and units simplify the
information-theoretic analysis. 

Theorem~\ref{thm:main} and
Corollary~\ref{cor:inf} are proved in Section~\ref{sec:proof1}. The
main ingredient here is an appropriate partitioning of the \dofs\ in a suitable
coordinate system, and the proof that the achievable rate of one group of \dofs\ is bounded in input. No assumption is made on
input power in this first proof. 

Theorem~\ref{thm:main} is proved again in Section~\ref{sec:proof2} by
considering the limit $\const P\rightarrow\infty$, 
which adds further intuition. Firstly, it is shown that, as the input
$\abs{\vect X}\rightarrow\infty$, the action of the discrete periodic stochastic nonlinear
Schr\"odinger (NLS) equation tends to multiplication by a random
matrix (with fixed probability distribution function (PDF), independent of the input). As a result,  perhaps counter-intuitively, as $\abs{\vect X}\rightarrow\infty$
noise simplifies the nonlinear channel to a 
\emph{linear} multiple-input multiple-output (non-coherent) fading
channel. Secondly, the asymptotic capacity is computed, without
calculating the conditional PDF of the channel, entropies, or solving the capacity optimization
problem. Because of the multiplicative noise, the asymptotic rate
depends only on the knowledge that whether channel random operator
has any deterministic component. The conditional PDF merely modifies the
bounded number $c$ in the Theorem~\ref{thm:main}.  

Note that we do not apply local analysis based on perturbation theories (valid in
the low power regime). The proof of the Theorem~\ref{thm:main}, \eg,
the asymptotic loss of \dofs, is based on a
global analysis valid for any signal and noise; see Section~\ref{sec:proof1}.

%%%%%%%%
% SECTION II: Notation and Preliminaries 
%%%%%%%

\section{Notation and Preliminaries}
\label{sec:notation}

The notation in this paper is motivated by \cite{moser2004dbb}.
Upper- and lower-case letters represent scalar random variables and
their realizations, \eg, $X$ and $x$. The same rule
is applied to vectors, which are distinguished using underline, \eg,
$\vect X$ for a random vector and $\vect x$ for a deterministic
vector. Deterministic matrices are shown by
upper-case letter with a special font, \eg, $\matd R=(r_{ij})$.
Random matrices are denoted by upper-case letters
with another special font, \eg, $\matr M=(M_{ij})$. 
Important scalars are distinguished with calligraphic font, \eg, $\const P$ for power and 
$\const C$ for capacity. The field of real and complex numbers is
respectively $\Reals$ and $\Complex$.

A sequence of numbers $X_1,\cdots, X_n $ is sometimes abbreviated as $X^n$, $X^0=\emptyset$.
A zero-mean circularly-symmetric complex Gaussian random vector with covariance 
matrix $\matd K$ is indicated by $\normalc{0}{\matd K}$. Uniform distribution on interval 
$[a,b)$ is designated as $\mathcal U(a,b)$.

Throughout the paper, the asymptotic equivalence $\const C(\const P) \sim
f(\const P)$, often abbreviated by saying
``asymptotically,'' means that $\lim_{\const P\rightarrow\infty}
\const C(\const P)/f(\const P)=1$. Letter $c\eqdef c(n,\const P)$ is reserved
to denote a real number bounded in $n$ and $\const P$. A sequence of
independent and identically distributed (\iid) random variables $X_n$ drawn from the PDF $p_X(x)$ is presented as $X_n\sim\iid\ p_X(x)$. 
The identity matrix with size $n$ is $I_n$.

The Euclidean norm of a vector $\vect x\in\Complex^n$ is
\begin{IEEEeqnarray*}{rCl}
\abs{x}=\left(|x_1|^2+\cdots+|x_n|^2\right)^{\frac{1}{2}}.
\end{IEEEeqnarray*}
This gives rise to an induced norm $\abs{\matd M}$ for matrix $\matd M$. 
We use the spherical coordinate system in the paper. Here, a vector $\vect x\in\Complex^n$ is 
represented by its norm $\abs{\vect x}$ and direction $\hat{\vect
  x}=\vect x/\abs{\vect x}$ (with
convention $\hat{\vect x}=0$ if $\vect x=0$). The direction can be
described by $m=2n-1$ angles.

When direction is random, its entropy  can be measured with respect to the 
spherical measure $\sigma^{m}(A)$, $A\subseteq \mathcal S^m$, where
$\mathcal S^m$ is the $m-$sphere
\begin{IEEEeqnarray*}{rCl}
\mathcal S^m=\left\{\hat{ \vect x}\in\Reals^{m+1}: ~\abs{\hat{\vect x}}=1 \right\}.
\end{IEEEeqnarray*}
It is shown in the Appendix~\ref{app:one} that the differential
entropy with respect to the Lebesgue and 
spherical measures, denoted  respectively by $h(\hat{\vect X})$ and $h_{\sigma}(\hat{\vect X})$, are related 
as
\begin{IEEEeqnarray}{rCl}
h(\vect X)=h(\abs {\vect X})+h_{\sigma}(\hat{\vect X}| \abs{\vect X})+m\E\log |X|.
\label{eq:sph-leb}
\end{IEEEeqnarray}
The entropy power of a random direction $\hat{\vect X}\in\Complex^n$ is
\begin{IEEEeqnarray*}{rCl}
  \vol(\hat{\vect X})=\frac{1}{2\pi e}\exp\bigl(\frac{2}{m}h_{\sigma}(\hat{\vect X})\bigr).
\end{IEEEeqnarray*} 
It represents the effective area of the support of  $\hat{\vect X}$
on $\mathcal S^m$.

%%%%%%%%
% SECTION I: Modified Split-Step Fourier Method
%%%%%%%

\section{The Modified Split-Step Fourier Method}
\label{sec:mssfm}
Signal propagation in optical fiber is described by the
stochastic nonlinear Schr\"odinger (NLS) equation \cite[Eq. 2]{yousefi2012nft1}
\begin{IEEEeqnarray}{rCl}
\frac{\partial Q}{\partial z}=L_L(Q)+L_N(Q)+N(t,z),
\label{eq:nls}
\end{IEEEeqnarray}
where $Q(t,z)$ is the complex envelope of the signal as a function of time
$t\in\Reals$ and space $z\in\Reals^+$ and $N(t,z)$ is zero-mean circularly-symmetric 
complex Gaussian noise with 
\begin{IEEEeqnarray*}{rCl}
  \E \bigl(N(t,z)N^*(t',z')\bigr)=\sigma^2\delta_{\const W}(t-t')\delta(z-z'),
\end{IEEEeqnarray*}
where $\delta_{\const W}(x)\eqdef 2\const W\sinc(2\const W x)$,
$\sinc(x)\eqdef \sin(\pi x)/(\pi x)$, and $\const W$ is noise bandwidth. The operator $L_L$ represents linear effects 
\begin{IEEEeqnarray}{rCl}
  L_L(Q)= \sum\limits_{k=0}^\infty
j^{k+1} \frac{\beta_k}{k!}\frac{\partial^k Q}{\partial t^k}
-\frac{1}{2}\alpha_{r}(t,z)\convolution Q(t,z), 
\label{eq:L-L}
\end{IEEEeqnarray}
where $\beta_k$ are dispersion coefficients, $\convolution$ is
convolution and
$\alpha_{r}$ is 
the residual fiber loss. The operator $L_N(Q)=j\gamma |Q|^2Q$ represents Kerr nonlinearity, where
 $\gamma$ is the nonlinearity parameter. The average power of 
the transmit signal is   
\begin{IEEEeqnarray}{rCl}
\const P= \lim\limits_{\mathcal T\rightarrow\infty}\E\frac{1}{\mathcal T}\int\limits_{-\mathcal T/2}^{\mathcal T/2}|Q(t,0)|^2\der t.
\label{eq:power-cont}
\end{IEEEeqnarray}

\begin{definition}[Loss Models]
The residual loss in \eqref{eq:L-L} accounts for uncompensated loss
and non-flat 
gain of the Raman amplification in distance and is
generally frequency dependent. The constant loss model refers to the case where
$\alpha_{r}(t,z)$ is constant in
the frequency $f$, \ie, 
$\hat{\alpha}_{r}(f,z)=\mathcal F(\alpha(t,z))\eqdef\alpha_{r}(z)$,
where $\mathcal F$ is the Fourier transform with respect to $t$. In
realistic systems, however, loss varies over frequency,
polarization or spatial models. This is the non-constant loss
model. Channel filters act similar to a non-constant loss function.
\label{def:loss}
\qed
\end{definition}

We discretize \eqref{eq:nls} in space and time. Divide a fiber of length $\const L$ 
into a cascade of a large number $m\rightarrow\infty$ of pieces of discrete fiber segments 
of length $\epsilon=\const L/m$ \cite[Section III. A]{yousefi2011opc}. A small
segment can be discretized in time and modeled in several ways. An appropriate approach is
given by the split-step Fourier method (SSFM).

The standard SSFM splits the \emph{deterministic} NLS equation into
linear and nonlinear parts. In applying SSFM to the \emph{stochastic} NLS equation, typically
noise is added to the signal. We introduce  a modified split-step Fourier method
where, instead of noise addition, the nonlinear part of
\eqref{eq:nls} is solved in the presence of noise analytically. 

In the linear step,  \eqref{eq:nls}  is solved with $L_N+N=0$. In the
discrete-time model,  linear step in a segment of length $\epsilon$
consists of multiplying a vector $\vect X\in\Complex^n$ by the 
dispersion-loss matrix $\matd R=(r_{kl})$.  In the constant loss model,
$\matd R=e^{-\frac{1}{2}\alpha_{r}\epsilon}\matd U$,
where $\matd U$ is a unitary matrix. In the absence of loss, $\matd R$ is
unitary. The values of 
$r_{kl}$ depend on the dispersion coefficients, $\epsilon$ and $n$. In general, all entries of $\matd R$ are non-zero, although in a small
segment, the off-diagonal elements can be very small. 

\begin{assumption}
Matrix $\matd R$ is fully dispersive, \ie, $r_{kl}\neq 0$, for all $k,l$.  
\label{ass:U}
\qed
\end{assumption}

In the nonlinear step, \eqref{eq:nls}  is solved with $L_L=0$
resulting in \cite[Eq. 12]{mecozzi1994llh}, \cite[Eq. 30]{yousefi2011opc}:
\begin{IEEEeqnarray}{rCl}
Q(t, z)=\left(Q(t,0)+W(t,z)\right)e^{j\Theta(t,z)},
\label{eq:zd}
\end{IEEEeqnarray}
in which
\begin{IEEEeqnarray*}{rCl}
\Theta(t,z)= \gamma\int\limits_0^z\Bigl|Q(t,0)+W(t,l)\Bigr|^2\der l,
\end{IEEEeqnarray*}
where $W(t,z)=\int_0^z N(t,l)\der l$ is Wiener process.
The modified nonlinear step in the MSSFM is obtained by discretizing \eqref{eq:zd}. Divide
a small segment $0\leq z\leq\epsilon$ into $L$ sub-segments of length
$\mu=\epsilon/L$. Define $\Phi:\Complex\times\Complex^n\mapsto [0,\infty)$ as
\begin{IEEEeqnarray}{rCl}
\Phi(X,\vect N)&=&\gamma\mu\underbrace{|X+N_1|^2+\gamma\mu|X+N_1+N_2|^2}_{\textnormal{signal-noise interactions, unknown}}
+\cdots \nn\\
&&+\gamma\mu\underbrace{|X+N_1+\cdots+N_L|^2}_{\textnormal{conditionally known}},
\label{eq:phase}
\end{IEEEeqnarray}
where $N_k\sim \iid\ \normalc{0}{\D/L}$, $\D=\sigma^2\const W\epsilon/n$.
The nonlinear step in a segment of length $\epsilon$ maps vector
$\vect X\in\Complex^n$ to vector $\vect Y\in\Complex^n$, according to 
\begin{IEEEeqnarray}{rCl}
 Y_k=\left(X_k+N_{k1}+\cdots+N_{kL}\right)e^{j\Theta(X_k, \vect N_k)},
\label{eq:discrete-zd}
\end{IEEEeqnarray}
where $\vect{N}_k=(N_{k1},\cdots, N_{kL})^T$, $N_{ki}\sim\iid\ \normalc{0}{\D/L}$.

The nonlinear step is a deterministic phase change in the SSFM. In this form, nonlinearity 
is entropy-preserving and does not interact 
with noise immediately  \cite[Lemma~2--3]{yousefi2015cwit2} --- unless several steps in the SSFM are
considered, which complicates the analysis. In the MSSFM, noise is introduced
in a distributed manner within each nonlinear step. This shows noise
influence more directly. 

Note that, conditioned on $\abs{Y_k}$, the last term in \eqref{eq:phase} is known. Other terms in \eqref{eq:phase} 
represent signal-noise interactions. They are conditionally unknown and are responsible for capacity 
limitation.

The MSSFM model for a fiber of length $\const L$ consists of the
cascade of linear and modified nonlinear steps (without noise addition between them).

\begin{definition}[Unit]
A \emph{unit} in the MSSFM model is defined as the cascade of three
segments of length $\epsilon$: A modified nonlinear step $\vect X\mapsto\vect U$, followed by a linear step
$\vect U\mapsto\vect V$, followed by another modified nonlinear step 
$\vect V\mapsto\vect Y$; see Fig.~\ref{fig:mssfm}. A unit of length $3\epsilon$ 
is the smallest piece of fiber whose capacity behaves qualitatively
similar to the capacity of the full model with length $\const L$.
\qed
\end{definition}

In the Appendix~\ref{app:in-out-mssfm} it is shown that the input
output relation $\vect X\mapsto \vect Y$ in one unit is given by
\begin{IEEEeqnarray}{rCl}
\vect Y=\matr M\vect X+\vect Z,\quad
\label{eq:one-seg}
\end{IEEEeqnarray}
where $\matr M\eqdef \matr M(\vect X,\matr N^1,\matr N^2)$ is a random matrix with entries
\begin{IEEEeqnarray}{rCl}
M_{kl}= r_{kl}e^{j\Phi_k+j\Psi_l},
\label{eq:Mkl}
\end{IEEEeqnarray}
in which
\begin{IEEEeqnarray*}{rCl}
\Psi_l=\Theta(X_l,\vect N_l^1),\quad \Phi_k=\Theta(V_k, \vect N_k^2).
\end{IEEEeqnarray*}
Here $\matr N^1=(\vect N_1^1,\cdots, \vect N_n^1)^T$ and $\matr
N^2=(\vect N_1^2,\cdots, \vect N_n^2)^T$ are $n\times L$ Gaussian ensembles with
\iid\ entries drawn from $\normalc{0}{\D/L}$, independent
of any other random variable. The additive noise $\vect Z\eqdef \vect Z(\vect X,\matr N^1,\matr N^2)$ is in
general non-Gaussian but bounded in $\abs{\vect X}$; see \eqref{eq:noise}. Finally, vector $\vect V$ is the output of
the linear step in Fig.~\ref{fig:mssfm}.

The input output relation $\vect X\mapsto \vect Y$ in a fiber of
length $\const L$ is obtained by composing $\bar m=m/2$ blocks $\vect Y_k=\matd
R\bigl(\matr M_k\vect X_k+\vect Z_k\bigl)$: 
\begin{IEEEeqnarray}{rCl}
\vect Y(k)=\matr M\bigl(k\bigr)\vect X(k)+\vect Z(k),
\label{eq:m-seg}
\end{IEEEeqnarray}
where $k=1,2,\cdots,$ is the transmission index, $\{\vect Z(k)\}_{k}$
is an \iid\ stochastic process, and 
\begin{IEEEeqnarray}{rCl}
\matr M(k)= \prod\limits_{k=1}^{\bar m} \matd R\matr M_k,~ \vect
Z(k)=\matd R\vect Z_{\bar m}+\sum\limits_{k=1}^{\bar m-1}
\Bigl(\prod\limits_{l=k+1}^{\bar{m}} \matd R\matr M_l\Bigr)\matd
R\vect Z_k.
\IEEEeqnarraynumspace
\label{eq:M-m-seg}
\end{IEEEeqnarray}
The power constraint \eqref{eq:power-cont} is discretized to
$\const P=\frac{1}{n}\E\norm{\vect X}^2$ in the discrete-time model.

\begin{remark}[Bandwidth Assumption]
Bandwidth, spectral broadening and spectral efficiency in the
continuous-time model are discussed
in Section~\ref{sec:cor}.
\qed
\end{remark}

\begin{remark}[Nonlinearity]
Note that $\matr M(\vect X,\matr N^1,\matr N^2)$ is a nonlinear
random operator. Particularly, it depends on
input. 
\qed
\end{remark}

\begin{remark}[Signal Dimension]
Dimension of the input space is $n$. 
To approximate the continuous-time model, $n\rightarrow\infty$. However,  
we let $n$ be arbitrary, \eg, $n=5$. Dimension should not be
confused with codeword length that tends to infinity.
\qed
\end{remark}

\begin{figure}
\centering{
\includegraphics{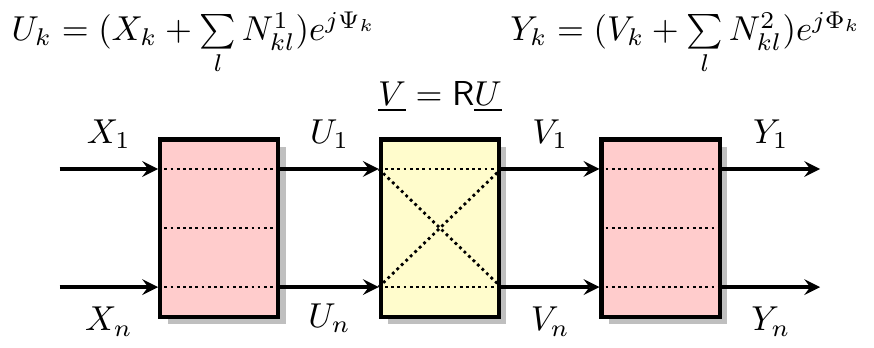}
}
\caption{One unit in the MSSFM. }
\label{fig:mssfm}
\end{figure}

%%%%%%%
% SECTION IV: Proof of the Theorem 1
%%%%%%%

\section{Proof of the Theorem~\ref{thm:main}}
\label{sec:proof1}

We first illustrate the main ideas of the proof via elementary examples.
 
Consider the additive white Gaussian noise (AWGN) channel $Y=X+Z$, where $X\in\Complex$ is input, $Y\in\Complex$ is output 
and $Z\sim\normalc{0}{1}$ is noise. Applying chain rule to the mutual
information
\begin{IEEEeqnarray*}{rCl}
  I(X;Y)=I(X;\abs Y)+I(X;\angle Y|\abs Y),
\end{IEEEeqnarray*}
where $\angle$ denotes phase. The amplitude channel $X\mapsto |Y|$ is 
\[
\abs{Y}\approx \abs{X}+Z_r,
\]
where $Z_r\sim\normalc{0}{\frac{1}{2}}$ and $\abs{X}\gg 1$. It asymptotically
contributes 
\begin{IEEEeqnarray*}{rCl}
  I(X; |Y|)\rightarrow \frac{1}{2}\log\const P+c
\end{IEEEeqnarray*}
to the capacity.

Phase, on the 
other hand, is supported on the finite interval $[0,2\pi)$. The only way that the contribution of the phase to 
the capacity could tend to infinity is that, phase noise tends to zero on the circle as $\abs X\rightarrow\infty$. Indeed,  
\begin{IEEEeqnarray*}{rCl}
\angle Y&=&\angle X+\tan^{-1}\left(\frac{Z_i}{\abs{X}+Z_r}\right)\\
&\approx&\angle X+\frac{Z_i}{\abs{X}},
\end{IEEEeqnarray*}
where $Z_r, Z_i\sim\iid\ \normalc{0}{\frac{1}{2}}$. The output entropy is clearly bounded, $h(\angle Y|\abs Y)\leq \log 2\pi$. However,
\begin{IEEEeqnarray}{rCl}
h(\angle Y|X, \abs Y)&=&h(Z_i)-\E\log|X|\nn\\
&\rightarrow&-\frac{1}{2}\log\const P+c,\quad \textnormal{as}\quad\const
P\rightarrow\infty.
\label{eq:cond-ent-awgn}
\end{IEEEeqnarray}
Note that the differential entropy can be negative. The contribution of the phase to the mutual information is 
\begin{IEEEeqnarray*}{rCl}
  I(X;\angle Y|\abs Y)\rightarrow \frac{1}{2}\log\const P+c'.
\end{IEEEeqnarray*}
Condition \eqref{eq:cond-ent-awgn} implies
 $\vol(\angle Y | X, \abs Y)\rightarrow 0$, \ie, the effective phase
 noise on the unit circle asymptotically vanishes.
 
Now consider the fading channel $Y=MX+Z$, where $X\in\Complex$ is
input, $Y\in\Complex$ is output and
$M, Z\sim\iid\ \normalc{0}{1}$. To prepare for generalization to 
optical channel, we represent a complex scalar $X$ as $\vect X=(\Re X,
\Im X)^T$. Thus
$\vect Y=\matr M\vect X+\vect Z$, where 
\begin{IEEEeqnarray*}{rCl}
\matr{M}=
\begin{pmatrix}
M_r & -M_i\\
M_i & M_r
\end{pmatrix},\quad 
\vect Z=
\begin{pmatrix}
Z_r\\
Z_i
\end{pmatrix}
,
\end{IEEEeqnarray*}
in which $M_{r,i},Z_{r,i}\sim\iid\ \normalc{0}{\frac{1}{2}}$. As $\abs{ \vect X}\rightarrow\infty$, $\vect Y\approx \matr{M}\vect X$, $\hat{\vect Y}\approx
\matr{M}\hat{\vect X}/\bigl|\matr M\hat{\vect X}\bigr|$, and randomness in 
$\hat{\vect Y}$  does not vanish with $\abs{X}$. Formally, 
\begin{IEEEeqnarray}{rCl}
h_{\sigma}(\hat{\vect{Y}}| \vect X, \abs{\vect Y})&=&
h_{\sigma}(\hat{\vect{Y}}\bigl| \matr M^{-1}\vect Y,\abs{\vect Y})
\nn\\
&=&
h_{\sigma}(\hat{\vect{Y}}| \matr M^{-1}\hat{\vect Y}, \abs{\vect Y})
\nn\\
&>&-\infty,
\label{eq:cond-entr-fad}
\end{IEEEeqnarray}
where \eqref{eq:cond-entr-fad} follows because $\vect a=\matr M^{-1}\hat{\vect Y}$ does not
determine $\hat{\vect Y}$ for random $\matr M$: There are four random
variables $M_{r,i}$ and $\hat{\vect Y}_{1,2}$ for three equations
$\matr M^{-1}\hat{\vect Y}=\vect a$ and $|\hat{\vect Y}|=1$. As a result, $I(\vect X;\hat{\vect Y}|\abs{\vect Y})<\infty$, and
$\abs{\vect Y}$ is the only useful \dof\ at high powers,
in the sense that its contribution $I(\vect X; \abs{\vect Y})$ to the
mutual information $I(\vect X; \vect Y)$ tends to infinity with
$\abs{\vect X}$.

The zero-dispersion optical fiber channel \eqref{eq:zd} is similar to the fading channel at high powers. 
The trivial condition
\begin{IEEEeqnarray*}{rCl}
h(\Theta(.,z) \Bigl| Q(.,0), |Q(.,z)|)>-\infty, \quad\forall Q(.,0),
\end{IEEEeqnarray*}
is sufficient to prove that the capacity of \eqref{eq:zd} is asymptotically the capacity of the amplitude channel, namely
$\frac{1}{2}\log\const P+c$. 

The intuition from the AWGN, fading and zero-dispersion channels
suggests to look at the dispersive optical channel in the
spherical coordinate system. The mutual information can be decomposed
using the chain rule
\begin{IEEEeqnarray}{rCl}
I(Q(0); Q(z))&=&I(\abs{Q(0)} ; Q(z))+I(\hat{Q}(0); Q(z)|\abs{Q(0)})\nn\\
&=&I(\abs{Q(0)} ; \abs{Q(z)})+I(\abs{Q(0)} ; \hat{Q}(z)\bigr|\abs{Q(z)})\nn\\
&&+I(\hat{Q}(0); Q(z)|\abs{Q(0)}),
\label{eq:I3}
\end{IEEEeqnarray}
where we dropped time index in $Q(t,z)$.

The first term in \eqref{eq:I3} is the rate of a single-input
single-output channel which can be computed in the asymptotic limit as follows. Let 
$\vect X$ and $\vect Y$ represent discretizations of the input $Q(0,.)$ and
output $Q(z,.)$. Consider first the lossless model. In this case, $\matr M$ is
unitary and from \eqref{eq:m-seg}, \eqref{eq:M-m-seg} and \eqref{eq:noise}
\begin{IEEEeqnarray}{rCl}
\abs{\vect Y}^2&=&\abs{\matr M\vect X+\vect Z}^2\nn\\
&=&\abs{\vect X+\matr M^\dag\vect Z}^2\nn\\
&=&\abs{\vect X+\vect Z}^2,
\label{eq:chi-squared}
\end{IEEEeqnarray}
where $\matr M^\dag$ is the adjoint (nonlinear) operator and
\eqref{eq:chi-squared} follows because $\vect Z$ and $\matr
M^\dag\vect Z$ are identically 
distributed when $\vect Z\sim\normalc{0}{m\D I_n}$; see Appendix~\ref{app:in-out-mssfm}.
Thus $|\vect Y|^2/(m\D)$ is a non-central chi-square random variable with $2n$ degrees-of-freedom and 
parameter $\abs{\vect x}^2/(m\D)$. The non-central chi-square conditional 
PDF $p(|\vect y|^2||\vect x|^2)$ can be
approximated at large $\abs{\vect x}^2$ using the Gaussian PDF, 
giving the asymptotic rate 
\begin{IEEEeqnarray}{rCl}
I(\abs{\vect X}; \abs{\vect Y})\rightarrow\frac{1}{2}\log\const P+c.
\label{eq:I(|X|;|Y|)} 
\end{IEEEeqnarray}
The bounded number $c$ can be computed using the exact PDF.

The case $\alpha_{r}(z)\neq 0$ is similar to the lossless case. Here $\matr
M=e^{-\frac{1}{2}\alpha_{r}\const L}\matr U$, where $\matr U$ is
a random unitary operator. Thus, $\matr
M^\dag=e^{-\frac{1}{2}\alpha_{r}\const L}\matr U^\dag$; furthermore  
$\abs{\matr M}=e^{-\frac{1}{2}\alpha_r\const L}$ is deterministic. The loss simply
influences the signal power, modifying constant $c$ in \eqref{eq:I(|X|;|Y|)}.

In the non-constant loss model, loss interacts with nonlinearity,
dispersion and noise. Here, $\abs{\matr M}$ is a random variable, and  
\begin{IEEEeqnarray}{rCl}
\abs{\vect Y}=\abs{\vect X}\abs{\matr M}\Bigl|\hat{\matr M}\hat{\vect
    X}+\frac{\vect Z}{\abs{\matr M}\abs{\vect X}}\Bigr|,
 \label{eq:multi-chan} 
\end{IEEEeqnarray}
where $\hat{\matr M}=\matr M/\abs{\matr M}$. Taking logarithm
\begin{IEEEeqnarray}{rCl}
\log\abs{\vect Y}=\log\abs{\vect X}+\log\abs{\matr M}+\log \Bigl|\hat{\matr M}\hat{\vect X}+\frac{\vect Z}{\abs{\matr M}\abs{\vect X}}\Bigr|.
\label{eq:log-fading}
\end{IEEEeqnarray}
Applying Lemma~\ref{lemm:decomposition}, we can assume $\abs
X>x_{0}$ for a suitable $x_{0}>0$ without changing the asymptotic capacity. The last term in
\eqref{eq:log-fading} is a bounded real random variable because 
\begin{IEEEeqnarray*}{rCl}
\sup_{\abs{\hat{\vect x}}=1, \abs{\hat{\matr M}}=1} \E \Bigl|\hat{\matr M}\hat{\vect x}+\frac{\vect Z}{\abs{\matr M}\abs{\vect 
X}}\Bigr|^2<\infty.
\end{IEEEeqnarray*}
Thus, the logarithm transforms the channel \eqref{eq:multi-chan} with multiplicative noise
$\abs{\matr M}$ to the channel \eqref{eq:log-fading} with additive bounded noise. The asymptotic capacity, independent of the PDF of $\abs{\matr M}$, is
\begin{IEEEeqnarray*}{rCl}
I(\abs{\vect X}; \abs{\vect Y})&\rightarrow& \frac{1}{2}\log
\E\left(\log(\abs{X})\right)^2+c\\
&=&\log\log\const P+c'.
\end{IEEEeqnarray*}

The last two terms in \eqref{eq:I3} are upper bounded in one unit of
the MSSFM using the data processing inequality 
\begin{IEEEeqnarray}{rCl}
I\left(\hat Q(0); Q(z)\bigl|\abs{Q(0)}\right)&\leq& I\left(\hat Q(0); Q(3\epsilon)\bigl|\abs{Q(0)}\right),
\label{eq:dp1}
\\
I\left(\abs{Q(0)}; \hat{Q}(z)\bigl| \abs{Q(z)}\right)&\leq&
I\left(\abs{Q(0)}; \hat Q(3\epsilon)\bigl|\abs{Q(3\epsilon)}\right).
\IEEEeqnarraynumspace
\label{eq:dp2}
\end{IEEEeqnarray}
We prove that the upper bounds in \eqref{eq:dp1}--\eqref{eq:dp2} do not scale with input $\abs{Q(0)}$. 

Let $\vect X,\vect Y\in\Complex^n$ denote discretization of $Q(0,t)$ and $Q(3\epsilon,t)$. 

\begin{lemma}
In one unit of the MSSFM
\begin{IEEEeqnarray}{rCl}
\IEEEyesnumber
\sup\limits_{\abs{\vect x}}\frac{1}{n}I(\hat{\vect X}; \vect Y\bigr|\abs{\vect x})&<&\infty,
\IEEEyessubnumber
\label{eq:hatX-Y}
\\
\sup\limits_{\abs{\vect y}}\frac{1}{n}I(\abs{\vect X}; \hat{\vect
  Y}\bigr|\abs{\vect y})&<&\infty.
\IEEEyessubnumber
\label{eq:absX-hatY}
\end{IEEEeqnarray}

\end{lemma}

\begin{figure}
\centering
\includegraphics[width=0.2\textwidth]{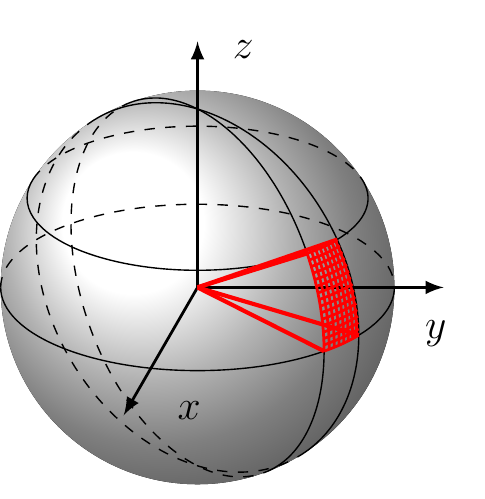}
\caption{The area on the surface of the unit sphere, representing $\vol(\hat Y)$, does not vanish as
$\abs{\underline{x}}\rightarrow\infty$.}
\label{fig:spherical-sector}
\end{figure}

\begin{proof}
Consider first the lossless model, where $\matr M$ is a unitary
operator. From Lemma~\ref{lemm:decomposition}, as $\abs{\vect
  x}\rightarrow\infty$, the additive noise in \eqref{eq:one-seg} can
be ignored. Thus $\vect Y=\abs{\vect Y}\hat{\vect Y}\approx\abs{\vect X}\hat{\vect Y}$. To prove \eqref{eq:hatX-Y},
\begin{IEEEeqnarray*}{rCl}
I(\hat{\vect X}; \vect Y\bigr|\abs{\vect x})&=&I(\hat{\vect
  X};\abs{\vect X}\hat{\vect Y}\bigr|\abs{\vect x})\\
&\overset{(a)}{=}&
I(\hat{\vect X};\hat{\vect Y}\bigr|\abs{\vect x})\\
  &=&h_{\sigma}(\hat{\vect Y}\bigr|\abs{\vect x})-h_{\sigma}(\hat{\vect Y}\bigr|\hat{\vect
  X}, \abs{\vect x}).
\end{IEEEeqnarray*}
Step $(a)$ follows from the identity
\begin{IEEEeqnarray}{rCl}
I(\vect X; Z\vect Y|Z)=I(\vect X; \vect Y|Z),\quad Z\neq 0.
\label{eq:I-indentity}
\end{IEEEeqnarray}

We measure the entropy of $\hat Y$ with respect to the spherical 
probability measure $\sigma^{m}$, $m=2n-1$, on the surface of the unit
sphere $S^{m}$.
From the maximum entropy theorem (MET) for distributions with compact support,  
\begin{IEEEeqnarray*}{rCl}
\frac{1}{n}h_{\sigma}(\hat{\vect Y}|\abs{\vect x})\leq \frac{1}{n}\log A_n,
\end{IEEEeqnarray*}
where $A_n=2\pi^n/\Gamma(n)$ is the surface area of $S^m$, in which
$\Gamma(n)$ is the gamma function.

We next show that the conditional entropy $h_{\sigma}(\hat{\vect
  Y}|\hat{\vect X}, \abs{\vect x})$ does not
tend to $-\infty$ with $\abs{\vect x}$.   The volume of the spherical 
sector in Fig.~\ref{fig:spherical-sector} vanishes if and
only if the corresponding area on the surface of the sphere vanishes.  
This can be formalized using identity \eqref{eq:sph-leb}. Let $\vect
W=U\hat{\vect{Y}}$, where $U\sim \mathcal U(0,1)$ independent of
$\vect X$ and $\vect Y$.  From \eqref{eq:sph-leb}
\begin{IEEEeqnarray}{rCl}
  h_{\sigma}(\hat{Y}|\hat{\vect X}, \abs{\vect x})=h(\vect
  W|\hat{\vect X}, \abs{\vect x})-h(U)-m\E\log U.
\label{eq:interm}
\end{IEEEeqnarray}

Applying chain rule to the differential entropy 
\begin{IEEEeqnarray}{rCl}
\IEEEyesnumber
h(\vect W|.)&=&\sum\limits_{k=1}^{n} h(W_k|W^{k-1},.)
\nn
\\
&=&\sum\limits_{k=1}^{n} h(\angle W_k, \bigl| W^{k-1},.)
\IEEEyessubnumber
\label{eq:chain-rule1}
\\
&&+\sum\limits_{k=1}^{n} h(|W_k| \bigl|W^{k-1}, \angle W_{k},.),
\IEEEyessubnumber
\label{eq:chain-rule2}
\end{IEEEeqnarray}
where entropy is conditioned on $\abs{\vect x}$ and $\hat{\vect X}$.

For the phase entropies in \eqref{eq:chain-rule1}, note that, from
\eqref{eq:one-seg}--\eqref{eq:Mkl}, 
$\angle W_k=\angle Y_k$ contains
random variable $\Phi_k$ with finite entropy, which does not appear 
in $W^{k-1}$. Formally, 
\[\angle W_k=\Phi_k+F(\Psi^n,\hat{\vect x}),
\]
for some function $F$, which can be determined from
\eqref{eq:one-seg}--\eqref{eq:Mkl}. Thus
\begin{IEEEeqnarray}{rCl}
h\left(\angle W_k\bigl| W^{k-1}, .\right)&=&
h\left(\Phi_k+F(\Psi^n,\hat{\vect x})\bigl|W^{k-1}, .\right)\nn\\
&\overset{(a)}{\geq} &h\left(\Phi_k+F(\Psi^n,\hat{\vect
    x}))\bigl|W^{k-1},\Phi^{k-1},\Psi^{n}, U, .\right)\nn\\
&\overset{(b)}{=}&h\left(\Phi_k+F(\Psi^n,\hat{\vect
    x}))\bigl|\Phi^{k-1},\Psi^{n}, U, .\right)\nn\\
&\overset{(c)}{=}&
h(\Phi_k\bigl|\Phi^{k-1},\Psi^{n},.)\nn\\
&>&
-\infty.
\label{eq:phase-entropies}
\end{IEEEeqnarray}
Step $(a)$ follows from the rule that conditioning reduces the entropy. Step $(b)$
holds because $W^{k-1}$ is a function of $\{\Phi^{k-1}, \Psi^n, U\}$. Step $(c)$ follows because $\{\Psi^n, .\}$ determines
$F(\Psi^n,\hat{\vect{x}})$.

For the amplitude entropies in \eqref{eq:chain-rule2}, we explain the
argument for $n=3$:
 \begin{IEEEeqnarray}{rCl}
W_k =Ue^{j\Phi_k}&&\Bigl(r_{k1}\hat
x_1e^{j\Psi_1}+r_{k2}\hat x_2e^{j\Psi_2}+r_{k3}\hat x_3e^{j\Psi_3}\Bigr),
\IEEEeqnarraynumspace
 \label{eq:Ys}
\end{IEEEeqnarray}
where $1\leq k\leq 3$. Noise addition in \eqref{eq:one-seg} implies 
$\pr(\hat X_k= 0)=0$, $\forall k$; we thus assume $\hat x_k\neq 0$ for
all $k$.  It is clear that $h(\abs{W_1})>-\infty$. 

There are 5 random variables $U$, $\Phi_1$, $\Psi_{1,2,3}$ for
two amplitude and phase relations in the $W_1$ equation in \eqref{eq:Ys}. Given $W_1$ and
$\angle W_2$, there are 6 random variables and three equations. One
could, for instance, express $\Psi_{1,2,3}$ in terms of $U$ and
$\Phi_{1,2}$. This leaves free at least $U$ in $|W_2|$, giving
\begin{IEEEeqnarray*}{rCl}
h(|W_2|\bigr| W_1, \angle W_2,.)&\geq &h(U)+c\\
&>&-\infty.
\end{IEEEeqnarray*}
The last equation for $W_3$ adds one random variable $\Phi_3$ and one equation
for $\angle W_3$. Together with the equation for $\abs{W_2}$, the
number of free random variables, defined as the number of all random variables minus the number of
equations, is 2; thus
\begin{IEEEeqnarray*}{rCl}
h(|W_3|\bigr| W_1, W_2, \angle W_3,.)>-\infty.
\end{IEEEeqnarray*}

In a similar way, in general, there are $n+k+1$ random variables in
$W^k$ and $2k-1$ equations in $(W^{k-1}, \angle W_k)$, resulting in
$n-k+2\geq 2$ free random variables. Thus
\begin{IEEEeqnarray}{rCl}
h(|W_k|\bigr| W^{k-1}, \angle W_{k})>-\infty,\quad 1\leq k\leq n.
\label{eq:amp-entropies}
\end{IEEEeqnarray}
Substituting \eqref{eq:phase-entropies} and \eqref{eq:amp-entropies} into \eqref{eq:chain-rule1}--\eqref{eq:chain-rule2}, 
we obtain $h(\vect W|.)>-\infty$. Finally, from \eqref{eq:interm}
\begin{IEEEeqnarray}{rCl}
  h_{\sigma}(\hat{Y}|\hat{\vect X}, \abs{\vect x})>-\infty.
\end{IEEEeqnarray}

The proof for lossy models, and \eqref{eq:absX-hatY}, is similar. Loss
changes matrix $\matd R$, which has no influence on our approach to
proving the boundedness of
terms in \eqref{eq:hatX-Y}--\eqref{eq:absX-hatY}.

\end{proof}

The essence of the above proof is that, as $\abs{\vect x}\rightarrow\infty$, the
additive noise in \eqref{eq:one-seg} gets smaller relative to the
signal, but phase noise (and thus randomness in $\matr M$) does not decrease
with $\abs{\vect x}$. Furthermore, $\matr M$ has enough randomness, owing
to the mixing effect of the dispersion, so that all $2n-1$ angles representing
signal direction in the spherical coordinate system are random
variables that do not vanish with $\abs{\vect x}$.

\begin{remark}
For some special cases of the dispersion-loss matrix $\matd R$, it is possible
to obtain deterministic components in $\hat{\vect Y}$ as  $\abs{\vect
  x}\rightarrow\infty$. These are cases where mixing does not fully
occur, \eg, $\matd R=I_n$. In the MSSFM, 
however, $\matd R$ is arbitrary, due to, \eg, step size $\epsilon$.
 \qed
\end{remark}

\subsection{Proof of the Corollary~\ref{cor:inf}}
\label{sec:cor}
We fix the power constraint and let $n\rightarrow\infty$ in the definition
of the capacity. The logarithmic terms depending on $\const P$ in the
Theorem~\ref{thm:main} approach zero, so that $\const C<\infty$.
 
Consider now the continuous-time model \eqref{eq:nls}. We discretize the channel in the frequency
domain, according to the approach in \cite{yousefi2015cwit2}. As the time
duration
$\const T\rightarrow\infty$ in \cite[Section~II]{yousefi2015cwit2}, we
obtain a discrete-time model with infinite number of \dofs\ (Fourier modes) in any frequency
interval at $z=0$. Therefore, $\const C<\infty$ in the corresponding 
discrete-time periodic model.

It is shown in \cite[Section~VIII]{yousefi2011opc} that, because of the spectral
broadening, the capacity of the continuous-time model $\const
C_c$ can be strictly lower than the capacity of the discrete-time model $\const C_d$.
Since $\const C_c\leq \const C_d$, and $\const C_d<\infty$, we obtain
$\const C_c<\infty$. 

We do not quantify constant $c'$ in the continuous-time model, which
can be much lower than the constant $c$ in the discrete-time model, due to spectral
broadening (potentially, $c'(\infty,\infty)=0$). A crude estimate, based on the Carson bandwidth rule, is given in 
\cite[Section~VIII]{yousefi2011opc} for the zero-dispersion
channel. 

To summarize, SE is bounded in input power in the continuous-time
model with $n=\infty$ (with or 
without filtering). The extent of the
data rate loss due to the spectral broadening ($c'$ versus $c$) remains an open problem.

%%%%%%%
% SECTION V: Random Matrix Model and the Asymptotic Capacity
%%%%%%%

\section{Random Matrix Model and the Asymptotic Capacity}
\label{sec:proof2}
In this section it is shown that, as $\abs{\vect X}\rightarrow\infty$,  the action of the discrete-time periodic stochastic NLS
equation  tends to multiplication by a random matrix (with fixed PDF, independent
of the input). Noise simplifies the NLS channel to a \emph{linear}
multiple-input 
multiple-output non-coherent fading channel. This section also
proves Theorems~\ref{thm:main} in an alternative intuitive way. 

The approach is based on the following steps. 

\emph{Step 1)} In Section \ref{sec:decomposition}, the input signal space is partitioned into a bounded region 
$\mathcal R^-$  
and its complement $\mathcal R^+$. It is shown that the overall rate is the interpolation of rates 
achievable using signals in 
$\mathcal R^{\pm}$. Lemma~\ref{lemm:I<infty} is proved, 
showing that the contribution of $\mathcal R^-$ to the mutual information is bounded. 
Suitable regions $\mathcal R^{\pm}$ are  chosen for the subsequent use.

\emph{Step 2)} In Section \ref{sec:fading-model}, it is shown that for
all $q(t,0)\in \mathcal R^+$, the nonlinear operator 
$L_N=j\gamma|Q|^2Q$ is multiplication by a uniform phase random variable, \ie, 
\begin{IEEEeqnarray*}{rCl}
L_N(Q)= j\Theta(t,z) Q,\quad \forall t, z,
\end{IEEEeqnarray*}
where\footnote{Derivatives do not exist with \iid\ phase random
  variables. However, with finite bandwidth, there is non-zero 
  correlation time.}$\Theta(t,z)\sim\iid\ \mathcal U(0,2\pi)$. 
In other words, for input signals in $\mathcal R^+$ the stochastic NLS equation is 
a simple linear channel with additive and multiplicative noise
\begin{IEEEeqnarray}{rCl}
\frac{\partial Q}{\partial z}=L_L(Q)+j \Theta(t,z)Q+N(t,z).
\label{eq:multiplicative}
\end{IEEEeqnarray}
Discretizing \eqref{eq:multiplicative}, we obtain that optical fiber 
is a fading channel when input is in $\mathcal R^+$:
\begin{IEEEeqnarray}{rCl}
\vect Y&=&\matr{M}\vect X+\vect Z,\quad \frac{1}{n}\E\norm{\vect X}^2\leq  \const P,
\label{eq:Y=HX+N}
\end{IEEEeqnarray}
in which $\matr M$ is a random matrix of the form 
\begin{IEEEeqnarray}{rCl}
\matr M = \prod\limits_{k=1}^m \matd R\matr D_k,\quad \matr D_k=\diag(e^{j\Theta_{ki}}),
\label{eq:M-expr}
\end{IEEEeqnarray}
where $\Theta_{kl}\sim\iid\ \mathcal U(0,2\pi)$ and $\vect Z$ is noise
\begin{IEEEeqnarray*}{rCl}
  \vect Z=\sum\limits_{k=1}^{m}\Bigl(\prod\limits_{l=1}^k\matd R\matr
  D_l\Bigr)\vect Z_k,\quad \vect Z_k\sim\iid\ \normalc{0}{\D I_n}.
\end{IEEEeqnarray*}
In general, $\matr M$ and $\vect Z$ are non-Gaussian. However, in the
constant loss model, $\vect Z\sim\normalc{0}{\matd K}$ where
$\matd K=(\sigma^2\const W\const L_e/n)I_n$, $\const
L_e=(1-e^{-\alpha\const L})/\alpha$. Note that $\matr M$ and $\vect Z$
have fixed PDFs, independent of $\vect X$.

Summarizing, the channel law is 
\begin{IEEEeqnarray}{rCl}
  p(\vect y|\vect x)=
\begin{cases}
\textnormal{given by the NLS equation}, &  \vect x\in\mathcal R^-,\\
p(\matr M\vect x+\vect Z|\vect x), & \vect x\in\mathcal R^+.
\end{cases}
\label{eq:law}
\end{IEEEeqnarray}

\emph{Step 3)} In Section~\ref{sec:asymptotic-capacity}, the capacity of the multiplicative-noise channel 
\eqref{eq:Y=HX+N} is studied. Lemma~\ref{lem:cap-Y=HX+N} and
\ref{lem:h(Mx)} are proved showing that,
for any $\matr M$ that does not have a deterministic component and is
finite (see \eqref{eq:h(M)>-infty}), the asymptotic capacity is 
given by the Theorem~\ref{thm:main}.
Importantly, the asymptotic rate is nearly independent of the
PDF of $\matr M$, which impacts only the bounded number
$c$. Finally, Lemma~\ref{lem:h(M-fiber)>-infty} is proved showing that the random matrix
underlying the optical fiber at high powers meets the assumptions of
the Lemma~\ref{lem:cap-Y=HX+N}.  An expression is provided for $c$,
which can be evaluated, depending on the PDF of $\matr M$.

\subsection{Step 1): Rate Interpolation}
\label{sec:decomposition}

We begin by proving the following lemma, which is similar to the proof approach 
in \cite{agrell2015conds}, where the notion of satellite 
constellation is introduced. 

\begin{lemma}
Let $p(\vect y|\vect x)$, $\vect x, \vect y \in\Reals^n$, be a conditional PDF. Define
\begin{IEEEeqnarray*}{rCl}
\vect X=\begin{cases}
\vect X_1, & \textnormal{with probability }\lambda,\\
\vect X_2, & \textnormal{with probability }1-\lambda,
\end{cases}
\end{IEEEeqnarray*}
where $\vect X_{1}$ and  $\vect X_2$ are random variables in $\Reals^n$ and $0\leq\lambda\leq 1$. Then
\begin{IEEEeqnarray}{rCl}
\lambda R_1+(1-\lambda)R_2\leq R\leq \lambda R_1+(1-\lambda)R_2+H(\lambda),
\label{eq:R1-R2-H}
\end{IEEEeqnarray}
where $R_1$, $R_2$ and $R$ are, respectively, mutual information of $X_1$, $X_2$ and $X$, and $H(x)=-x\log x-(1-x)\log(1-x)$ is the
binary entropy function, $0\leq x\leq 1$. 
\label{lemm:decomposition}
\end{lemma}

\begin{proof}
The PDF of the time sharing random variable $\vect X$ and its output $\vect Y$ are
\begin{IEEEeqnarray}{rCl}
p_{\vect X}(\vect x)&=&\lambda p_{\vect X_1}(\vect x)+(1-\lambda)p_{\vect X_2}(\vect x),\nn
\\
p_{\vect Y}(\vect y)&=&\lambda p_{\vect Y_1}(\vect y)+(1-\lambda)p_{\vect Y_2}(\vect y),
\label{eq:py=py1+py2}
\end{IEEEeqnarray}
where
\begin{IEEEeqnarray*}{rCl}
p_{\vect Y_1,\vect Y_2}(\vect y)=\int p(\vect y|\vect x)p_{\vect X_1,\vect X_2}(\vect x)\der \vect x.
\end{IEEEeqnarray*}
By elementary algebra
\begin{IEEEeqnarray*}{rCl}
I(\vect X; \vect Y)
= \lambda I(\vect X_1;\vect Y_1)+(1-\lambda) I(\vect X_2,\vect Y_2)+\Delta I,
\end{IEEEeqnarray*}
where
\begin{IEEEeqnarray*}{rCl}
\Delta I &=& 
\lambda D( p_{\vect Y_1}(\vect y_1)||p_{\vect Y}(\vect y))
+(1-\lambda) D( p_{\vect Y_2}(\vect y_2)||p_{\vect Y}(\vect y)).
\end{IEEEeqnarray*}

From \eqref{eq:py=py1+py2}, 
\[
p_{\vect Y}(\vect y)\geq \max\left\{\lambda p_{\vect Y_1}(\vect y), (1-\lambda) p_{\vect Y_2}(\vect 
y)\right\},
\]
which gives $\Delta I\leq H(\lambda)$. 
From $\log x\leq x-1$, $D(p(\vect y_{1,2})||p(\vect y))\geq 0$, giving
$\Delta I\geq 0$.
Thus
\begin{IEEEeqnarray*}{rCl}
I (\vect X;\vect Y) &\leq & \lambda I (\vect X_1;\vect Y_1)+(1-\lambda) I(\vect X_2;\vect Y_2)+H(\lambda),\\
 I(\vect X;\vect Y) &\geq& \lambda I (\vect X_1;\vect Y_1)+(1-\lambda) I(\vect X_2;\vect Y_2).
\end{IEEEeqnarray*}

\end{proof}

\begin{corollary}
Define
\begin{IEEEeqnarray*}{rCl}
\bar R=\lim\limits_{n\rightarrow\infty}\frac{1}{n}I(\vect X;\vect Y).
\end{IEEEeqnarray*}
With definitions in the Lemma~\ref{lemm:decomposition}, we have $\bar
R=\lambda\bar R_1+(1-\lambda)\bar R_2$.
\label{cor:rate-interpolation}
\end{corollary}

For the rest of the paper, we choose $\mathcal R^-$ to be an $n$-hypercube in $\Complex^n$
\begin{IEEEeqnarray*}{rCl}
\mathcal R^-_{\kappa}=\bigl\{\vect x\in\Complex^n~\Bigl|~|x_k|<
\kappa,\quad 1\leq k\leq n
\bigr\},
\end{IEEEeqnarray*}
and $\mathcal R^+_{\kappa}=\Complex^n\backslash\mathcal
R^-_{\kappa}$. We drop the subscript $\kappa$ when we do not need it. The following Lemma shows
that, if $\kappa<\infty$, the contribution of the signals in $\mathcal R^-_\kappa$
to the mutual information in the NLS channel is bounded.

\begin{lemma}
Let $\vect X\in\Complex^n$ be a random variable supported on $\mathcal
R^-_{\kappa}$ and $\kappa<\infty$. For the NLS channel \eqref{eq:nls}
\begin{IEEEeqnarray*}{rCl}
\frac{1}{n}I(\vect X;\vect Y)<\infty.
\end{IEEEeqnarray*}
\label{lemm:I<infty}
\end{lemma}
\begin{proof}
From the MET, $h(\vect Y)\leq \log\left|\mathcal R^-_\kappa\right|< \infty$. 
Let $\vect X\rightarrow \vect Z\rightarrow \vect Y=\vect Z+\vect N$ be a Markov chain, 
where $\vect N$ is independent of $\vect Z$ and $h(\vect N)>-\infty$. Then 
$h(\vect Y|\vect X)\geq h(\vect Y|\vect X,\vect Z)=h(\vect Y|\vect Z)=h(\vect N)>-\infty$. Applying this to the NLS 
channel with an independent noise addition in the last
stage, we obtain $I(\vect X, \vect Y)<\infty$. 

Alternatively, from \cite{yousefi2015cwit2}, 
\[
\frac{1}{n}I(\vect X; \vect Y)\leq \log(1+\frac{|\mathcal R^-_\kappa|^2}{nm\D })<\infty.
\]
\label{lem:R-}
\end{proof}

\subsection{Step 2): Channel Model in the High Power Regime}
\label{sec:fading-model}

We begin with the zero-dispersion channel. Let $Q(t,0)=X=R_x\exp(j\Phi_x)$ and $Q(t,z)=Y=R_y\exp(j\Phi_y)$ be, respectively, 
channel input and output in \eqref{eq:zd}. For a fixed $t$, $X$ and
$Y$ are complex numbers. 

\begin{lemma}
We have 
\begin{IEEEeqnarray*}{rCl}
\lim\limits_{\abs{x}\rightarrow\infty}p(\phi_y|x)&=&\lim\limits_{\abs{x}\rightarrow\infty}p(\phi_y|x,r_y)\\
&=&\frac{1}{2\pi}.
\end{IEEEeqnarray*}
Thus, the law of the zero-dispersion channel tends to the law of the following channel
\[
Y=Xe^{j\Theta}+Z,
\]
where $\Theta\sim\mathcal U(0,2\pi)$, $Z\sim \normalc{0}{\D}$, and $(X,Z,\Theta)$ are independent.
\label{lemm:uniform}
\qed
\end{lemma}

\begin{proof}
The condition PDF is \cite[Eq. 18]{yousefi2011opc}
\begin{IEEEeqnarray*}{rCl}
p(r_y,\phi_y|r_x,\phi_x)&=&\frac{1}{2\pi}p_0(r_y|r_x)\\
&&+\:\frac{1}{2\pi}\sum\limits_{m=1}^\infty\Re
\left(p_m(r_y|r_x)e^{jm(\phi_y-\phi_x-\gamma r^2_xz)}\right).
\end{IEEEeqnarray*}
Here 
\begin{IEEEeqnarray*}{rCl}
p_m(r_y|r_x)&=&2r_xb_m\exp\left(-a_m(r_x^2+r_y^2)\right)I_m(2b_mr_xr_y),
\end{IEEEeqnarray*}
where
\begin{IEEEeqnarray*}{rCl}
a_m=\frac{1}{\mathcal D z}x_m\coth(x_m),\quad b_m=\frac{1}{\mathcal Dz}\frac{x_m}{\sinh(x_m)},
\end{IEEEeqnarray*}
in which $x_m=\sqrt{jm\gamma\mathcal D}z=t_m(1+j)$, $t_m=\sqrt{\frac{1}{2}m\gamma\mathcal D}z$.
Note that $p(r_y|r_x)=p_0(r_y|r_x)$. 

The conditional PDF of the phase is
\begin{IEEEeqnarray*}{rCl}
p(\phi_y|r_x,\phi_x, r_y)&=&\frac{p(\phi_y, r_y|r_x,\phi_x)}{p(r_y|r_x,\phi_x)}
\\
&\overset{(a)}{=}&\frac{p(\phi_y, r_y|r_x,\phi_x)}{p(r_y|r_x)}
\\
&=&\frac{1}{2\pi} 
\sum\limits_{m=1}^\infty\Re\left( D_m(r_x)  e^{jm(\phi_y-\phi_x-\gamma
    r_x^2z)}\right)\\
&& +\frac{1}{2\pi},\nn
\end{IEEEeqnarray*}
where step $(a)$ follows from $p(r_y|r_x,\phi_x)=p(r_y|r_x)$ (see \cite[Fig.~6 (b)]{yousefi2011opc}) and
\begin{IEEEeqnarray}{rCl}
D_m(r_x)&=&\frac{p_m(r_y|r_x)}{p_0(r_y|r_x)}\nn\\
&=&\frac{b_m}{b_0}\frac{I_m(2b_mr_xr_y)}{I_0(2b_0r_xr_y)}\nn\\
&&\times
\exp\Bigl\{-b_0\left(x_m\coth x_m-1\right)(r_x^2+r_y^2)\Bigr\}.
\IEEEeqnarraynumspace 
\label{eq:Dm}
\end{IEEEeqnarray}

The following three inequalities can be verified:
\begin{IEEEeqnarray}{rCl}
\IEEEyesnumber
\left|\frac{b_m}{b_0}\right|^2&=&
\frac{4t^2}{\left|\cosh 2t-\cos 2t\right|}
\nn
\\
&\leq& 1,\quad t>0.
\IEEEyessubnumber
\label{eq:inq1}
\end{IEEEeqnarray}
\begin{IEEEeqnarray}{rCl}
\left|\frac{I_m(2r_xr_yb_m)}{I_0(2r_xr_y b_0)}\right|&\leq& 
\left|\frac{I_m(2r_xr_yb_0)}{I_0(2r_xr_y b_0)}\right|\nn\\ 
&\leq&1.
\IEEEyessubnumber
\label{eq:inq2}
\end{IEEEeqnarray}
\begin{IEEEeqnarray*}{rCl}
F(t)&=&\Re(x_m\coth x_m-1)\nn\\
&=&t\frac{\sinh(2t)+\sin(2t)}{\cosh(2t)- \cos(2t)}-1
\nn\\
&>&0,
\IEEEyessubnumber
\label{eq:inq3}
\end{IEEEeqnarray*}
where $t\eqdef t_m>0$. 

Using \eqref{eq:inq1}--\eqref{eq:inq3} in \eqref{eq:Dm}, we obtain
$|D_m(r_x)|\leq E_m(r_x)$, where
\begin{IEEEeqnarray}{rCl}
E_m(r_x)=\exp\left\{-\frac{1}{\mathcal
    Dz}F(t_m)(r_x^2+r_y^2)\right\}.
\label{eq:D<1}
\end{IEEEeqnarray}
We have
\begin{IEEEeqnarray*}{rCl}
\lim\limits_{r_x\rightarrow\infty}\left|\sum\limits_{m=1}^\infty D_m(r_x)  e^{jm(\phi_y-\phi_x-\gamma r_x^2z)}\right|
&\leq& 
\lim\limits_{r_x\rightarrow\infty}
\sum\limits_{m=1}^\infty
\left|D_m(r_x)\right| 
\\
&\leq&
\lim\limits_{r_x\rightarrow\infty}
\sum\limits_{m=1}^\infty
E_m(r_x) 
\\
&\overset{(a)}{=}&
\sum\limits_{m=1}^\infty
\lim\limits_{r_x\rightarrow\infty}
E_m(r_x)\\
&\overset{(b)}{=}& 0.
\end{IEEEeqnarray*}
Step $(a)$ follows because $E_m(r_x)\leq E_m(0)$ and
$\sum E_m(0)$ is convergent; thus, by the dominated convergence theorem, $\sum E_m(r_x)$ is
uniformly convergent. Step $(b)$ follows from \eqref{eq:D<1}.

It follows that
\begin{IEEEeqnarray*}{rCl}
\lim\limits_{r_x\rightarrow\infty}p(\phi_y|r_x, \phi_x, r_y)=\frac{1}{2\pi}. 
\end{IEEEeqnarray*}
Furthermore 
\begin{IEEEeqnarray*}{rCl}
\lim\limits_{r_x\rightarrow\infty}p(\phi_y|r_x,\phi_x)&=&
\lim\limits_{r_x\rightarrow\infty}\int
 p(\phi_y|r_x,\phi_x, r_{y'})p(r_{y'}|r_x,\phi)\der r_{y'}\\
&=&\frac{1}{2\pi}. 
\end{IEEEeqnarray*}

\end{proof}

Lemma~\ref{lemm:uniform} generalizes to the vectorial zero-dispersion channel \eqref{eq:zd}. 
Since noise is independent and identically distributed in space and time,
so are the corresponding uniform phases. This is true even if $X_i$ in Fig.~\ref{fig:mssfm}
are dependent, \eg, $\vect X=(x,\cdots, x)$.

We now consider the dispersive model. To generalize
Lemma~\ref{lemm:uniform} to the full model, we use the following
notion \cite[Section~2.6]{moser2004dbb}.

\begin{definition}[Distributions that Escape to Infinity]
A family of PDFs $\{p_{\vect X_{\theta}}(\vect
x)\}_{\theta}$, $0\leq\theta\leq\theta_0$, is said to \emph{escape to infinity}
with $\theta$ if $\lim\limits_{\theta\rightarrow\theta_0}\pr(|\vect X_{\theta}|<c)=0$ for
any finite $c$.
\qed
\end{definition}

\begin{lemma}
Let $\vect X\in\mathcal R^+_{\kappa}$ and $\vect Y$ be, respectively,
the channel input and output in the dispersive model. The PDF of
$Y_k$ escapes to infinity as $\kappa\rightarrow\infty$ for all $k$.
\label{lemm:scape}
\end{lemma}

\begin{proof}
The proof is based on induction in the MSSFM units. We make precise
the intuition that, as $\kappa\rightarrow \infty$, the PDF of $|Y_k|$
spreads out, so that an ever decreasing probability is assigned to any
finite interval.

Consider vector $\vect V$ in
Fig.~\ref{fig:mssfm}, at the end of the linear step in the first unit. Setting $W=\abs{V_k}$, we have
\begin{IEEEeqnarray}{rCl}
\pr(W<c)&=&\int\limits_{0}^{c} p_W(w)\der w\nn\\
&=&
\int\limits_{0}^{\epsilon c} \frac{1}{\epsilon}p_W(\frac{t}{\epsilon})\der t\nn\\
&\leq&
\epsilon c \norm{p_{T_\epsilon}(t)}_{\infty},
\label{eq:p(W<c)}
\end{IEEEeqnarray}
where $T_{\epsilon}=\epsilon W$, $p_{T_{\epsilon}}(t)=\frac{1}{\epsilon}p_W(\frac{t}{\epsilon})$ and
$\epsilon\eqdef 1/\kappa$. Below, we prove that $\norm{p_{T_0}(t)}_{\infty}<\infty$.

Fix $0<\delta<1$ and define the (non-empty) index set
\begin{IEEEeqnarray*}{rCl}
\mathcal I=\{i: |x_i|\geq \kappa^{1-\delta}\}.
\end{IEEEeqnarray*}
The scaled random variable $T_{\epsilon}$ is
\begin{IEEEeqnarray*}{rCl}
T_{\epsilon}&=&\epsilon|V_k|\\
 &=&\Bigl|\sum\limits_{l=1}^n e^{j\Psi_l(x_l,\vect{N}_l^1)}r_{kl}\tilde{x}_l+\epsilon\tilde{\vect{Z}}\Bigr|\\
&=&\Bigl|
\sum\limits_{l\in\mathcal I}
+
\sum\limits_{l\notin\mathcal I}+\epsilon\tilde{\vect Z} 
\Bigr|,
\end{IEEEeqnarray*}
where $\tilde{\vect Z}$ is an additive noise and $\tilde{\vect x}=\epsilon\vect
x$.

As $\epsilon\rightarrow 0$,
the second sum vanishes because, if $l\notin \mathcal I$, 
$|\tilde{x}_l|<\epsilon^\delta\rightarrow 0$. In the first sum,
$|x_l|\rightarrow\infty$, thus $\Psi_l(x_l,\vect{N}_l^1)\overset{\textnormal{a.s.}}{\rightarrow}
U_l$, where $U_l\sim\mathcal U(0,2\pi)$. Therefore
$T_{\epsilon}\overset{\textnormal{a.s.}}{\rightarrow} T_0$, in which
\begin{IEEEeqnarray}{rCl}
  T_{0}=\Bigl|\sum\limits_{l\in\mathcal
    I}e^{jU_l}r_{kl}\tilde x_l\Bigr|,
\label{eq:T-eps}
\end{IEEEeqnarray}
where $|\tilde x_l|>0$. 
Since the PDF of $e^{jU_l}$ is in
$L^{\infty}(\mathbb T)$ on the circle $\mathbb T$, so is the conditional PDF
$p_{T_0|\tilde{\vect{ X}}}(t|\tilde{\vect x})$, \ie,
\begin{IEEEeqnarray}{rCl}
\norm{p_{T_0}(t)}_{\infty}<\infty.
\label{eq:p(t)<infty}
\end{IEEEeqnarray}
 Substituting \eqref{eq:p(t)<infty} into \eqref{eq:p(W<c)}
\begin{IEEEeqnarray}{rCl}
\lim\limits_{\epsilon \rightarrow 0}\pr(W<c)=0.
\label{eq:p(W<c)-2}
\end{IEEEeqnarray}

In a similar way, \eqref{eq:p(W<c)-2} can be proved for $\vect V$ at the output of
the linear step in the second unit, by replacing $e^{jU_l}r_{kl}$
in \eqref{eq:T-eps} with $M_{kl}$, and noting that, 
as $\epsilon\rightarrow 0$, $\{M_{kl}\}_{l\in\mathcal I}$ tend to random variables
independent of input, with a smooth PDF (without delta functions). 

\end{proof}

From the Lemma~\ref{lemm:scape}, as $\kappa\rightarrow\infty$ the
probability distribution at the input of
every zero-dispersion segment in the link 
escapes to infinity, turning the operation
of the nonlinearity in that segment into multiplication by a uniform
phase and independent noise addition.  We thus obtain an input region $\mathcal R^+_\kappa$ for
which, if $\vect x\in\mathcal R^+_\kappa$, the channel is
multiplication by a random matrix, as described in
\eqref{eq:Y=HX+N}. The channel converts any small noise into
worst-case noise in evolution.

\subsection{Step 3): The Asymptotic Capacity}
\label{sec:asymptotic-capacity}

In this section, we obtain the asymptotic capacity of the channel \eqref{eq:law}.

Applying Lemma~\ref{lemm:decomposition} to \eqref{eq:law}
\begin{IEEEeqnarray*}{rCl}
  \bar R(\const P)=\lambda \bar R_-(\const P)+(1-\lambda)\bar
  R_+(\const P),
\end{IEEEeqnarray*}
where $\bar R_{\pm}(\const P)=\frac{1}{n}I(\vect X;\vect Y)$, $\vect X\in\mathcal
R^{\pm}_\kappa$  and  $\lambda$ is a parameter to be optimized. To
shorten the analysis, we ignore
the term $c=H(\lambda)/n$ in \eqref{eq:R1-R2-H}, as it does not depend
on $\const P$.

We choose $\kappa$ sufficiently large, \emph{independent of the average
input power}
$\const P$. From the Lemma~\ref{lemm:I<infty}, $\sup_{\const P}\bar R_-(\const
P)<\infty$. 
The following Lemma shows that $\bar{R}_+(\const P)$
is given by the logarithmic terms in Theorem~\ref{thm:main} with
\begin{IEEEeqnarray*}{rCl}
c\geq\lambda \bar R_-+(1-\lambda)\sup\limits_{\hat{\vect X}}\frac{1}{n}I(\hat{\vect X};\matr M\hat{\vect X}),  
\end{IEEEeqnarray*}
where $\bar R_-$ is the achievable rate at low powers. If $\matr M$ is Haar
distributed, $c=\lambda \bar R_-$.

Define $h(\matr M)\eqdef h(M_{11},\cdots, M_{nn})$. 

\begin{lemma}
Assume that 
\begin{IEEEeqnarray}{rCl}
h(\matr M)>-\infty, \quad \E|M_{ij}|^2<\infty,\quad 1\leq
i,j\leq n.
\label{eq:h(M)>-infty}
\end{IEEEeqnarray}
Then, the asymptotic capacity of \eqref{eq:Y=HX+N} is given by the
expressions stated in the Theorem~\ref{thm:main}.
\label{lem:cap-Y=HX+N}
\end{lemma}

\begin{proof}

The capacity of the multiple-input multiple-output non-coherent
memoryless fading channel \eqref{eq:Y=HX+N} is studied in
\cite{moser2004dbb,lapidoth2003capacity}. Here, we present
a short proof with a bit of approximation. 

Using chain rule for the mutual information
\begin{IEEEeqnarray}{rCl}
I(\vect X;\vect Y)&=&I(\abs{\vect X}; \vect Y)+I(\hat{\vect X}; \vect
Y\bigl|\abs{\vect X})\nn\\
&=&I(\abs{\vect X}; \abs{\vect Y})+I(\hat{\vect X}; \vect Y\bigl|\abs{\vect X})+I(\abs{\vect X}; \hat{\vect
  Y}|\abs{\vect Y}).
\IEEEeqnarraynumspace
\label{eq:I(X;Y)}
\end{IEEEeqnarray}

The first term in \eqref{eq:I(X;Y)} gives the logarithmic terms in Theorem~\ref{thm:main}, as
calculated in Section~\ref{sec:proof1}. We prove that the other
terms are bounded in $\abs{\vect X}$. From the Lemma~\ref{lemm:decomposition}, the additive noise in \eqref{eq:Y=HX+N} can be ignored when $\vect
X\in\mathcal R^+_\kappa$,  so that $\vect Y\approx\matr M\vect X$.

The second term in \eqref{eq:I(X;Y)} is
\begin{IEEEeqnarray*}{rCl}
I(\hat{\vect X}; \vect Y\bigl|\abs{\vect X}) &=&  I(\hat{\vect X};
|\vect X| \matr M\hat{\vect X}\bigl|\abs{\vect X})\\
&=&I(\hat{\vect
  X}; \matr M\hat{\vect X}\bigr| \abs{\vect X}),
\end{IEEEeqnarray*}
where we used identity \eqref{eq:I-indentity}. Note that we can not assume that $\abs{\vect X}$ and $\hat{\vect X}$ are
independent. 

For the output entropy 
\begin{IEEEeqnarray*}{rCl}
  h(\matr M\hat{\vect X}|\abs{\vect X})&\leq& h(\matr M\hat{\vect X})\\
 &\overset{(a)}{\leq}&\sum\limits_{k=1}^n
 h\Bigl(\sum\limits_{l=1}^n M_{kl}\hat{\vect X}_l \Bigr)
\\
  &\overset{(b)}{\leq}&\sum\limits_{k=1}^n \log\Bigl(\pi e
   \E\Bigl|\sum\limits_{l=1}^n M_{kl}\hat{\vect
     X}_l\Bigr|^2\Bigr)
\\
 &\overset{(c)}{\leq}&\sum\limits_{k=1}^n \log\Bigl(\sum\limits_{l=1}^n
   \E|M_{kl}|^2\Bigr)+n\pi e
\\
 &\overset{(d)}{\leq}&n\log\Bigl(\frac{1}{n}\E|\matr{M}|^2_F\Bigr)+n\pi e\\
&<&\infty,
\end{IEEEeqnarray*}
where $\abs{\matr M}_F=\Bigl(\sum\limits_{k,l=1}^n|M_{kl}|^2\Bigr)^{\frac{1}{2}}$ is the Frobenius norm. Step $(a)$ is
obtained using the inequality 
$h(\vect W)=\sum_k h(W_k|W^{k-1})\leq\sum_k h(W_k)$. 
Step $(b)$ is due to the MET. Cauchy-Schwarz and
Jensen's  inequalities are, respectively, applied in steps $(c)$ and $(d)$. 

For the conditional entropy
\begin{IEEEeqnarray}{rCl}
h(\matr M\hat{\vect   X}|\abs{\vect X}, \hat{\vect X})&=&
\E_{\vect{X}} h(\matr M\hat{\vect   x}|\abs{\vect x}, \hat{\vect x})
\nn\\
&\geq& \inf\limits_{\hat{\vect x}} h(\matr M\hat{\vect x})\nn\\
&\overset{(a)}{>}&-\infty. 
\label{eq:inter-6-aa}
\end{IEEEeqnarray}
Step $(a)$ holds because, from the Lemma~\ref{lem:h(Mx)}, $h(\matr
M\hat{\vect x})>-\infty$ for any $\hat{\vect x}$.

The third term in \eqref{eq:I(X;Y)} can be upper bounded using the second 
term by setting $\vect X=\matr{M}^{-1}\vect Y$. We prove it alternatively. Since $\hat{\vect Y}$ is compactly supported,
$h_{\sigma}(\hat{\vect Y}||\abs{\vect Y}|)\leq h_{\sigma}(\hat{\vect
  Y})<\infty$. The conditional entropy is
\begin{IEEEeqnarray}{rCl}
h_{\sigma}(\hat{\vect Y}\bigr|\abs{\vect X},\abs{\vect Y})&=&
h_{\sigma}(\hat{\vect Y}\bigr|\abs{\vect X},\abs{\vect X}|\matr M\hat{\vect X}|)
\nn\\
&=&h_{\sigma}\Bigl(\frac{\matr M\hat{\vect X}}{|\matr M\hat{\vect
    X}|}\Bigr|\abs{\vect X},|\matr M\hat{\vect X}|\Bigr).
\label{eq:cond-entropy-Y=HX+N}
\end{IEEEeqnarray}

Applying identity \eqref{eq:sph-leb} to $\matr M\hat{\vect{X}}$
and conditioning on $\abs{\vect X}$
\begin{IEEEeqnarray}{rCl}
h_{\sigma}\Bigl(\frac{\matr M\hat{\vect X}}{|\matr M\hat{\vect
    X}|}\Bigr|\abs{\vect X}, |\matr M\hat{\vect X}|\Bigr) &=& h(\matr
M\hat{\vect X}\bigr| \abs{\vect X})-h(|\matr M\hat{\vect X}|\bigr|
\abs{\vect X})\nn\\
&&
-(2n-1)\E\bigl(\log(|\matr M\hat{\vect X}|)\bigr|\abs{\vect X}\bigr).
\IEEEeqnarraynumspace
\label{eq:inter-6}
\end{IEEEeqnarray}
For the first term in \eqref{eq:inter-6} 
\begin{IEEEeqnarray}{rCl}
\IEEEyesnumber
h(\matr M\hat{\vect X}|\abs{\vect X})&\geq&h(\matr M\hat{\vect
  X}|\abs{\vect X}, \hat{\vect X})\nn\\
&>&-\infty, 
\IEEEyessubnumber
\label{eq:inter-6-a}
\end{IEEEeqnarray}
where we used \eqref{eq:inter-6-aa}.
Since $|\matr M\hat{\vect X}|\leq\abs{\matr M}\leq \abs{\matr M}_F$,
from the MET
\begin{IEEEeqnarray}{rCl}
h\bigl(|\matr M\hat{\vect X}|\bigr|\abs{\vect X}\bigr)&\leq&  
h\bigl(|\matr M\hat{\vect X}|\bigr)\nn\\
&\leq& \frac{1}{2}\log \bigl(2\pi e\E|\matr M\hat X|^2\bigr)\nn\\
&\leq& \frac{1}{2}\log \bigl(2\pi e\E\abs{\matr M}_F^2\bigr)\nn\\
&<&\infty.
\IEEEyessubnumber
\label{eq:inter-6-b}
\end{IEEEeqnarray}
Furthermore, 
\begin{IEEEeqnarray}{rCl}
  \E\bigl(\log(|\matr M\hat{\vect X}|^2)\bigr|\abs{\vect X}\bigr)&\leq&
  \log\E\bigl(|\matr M\hat{\vect X}|^2\bigr|\abs{\vect X}\bigr)
\nn\\
 &\leq&\log\E\abs{\matr M}_F^2\nn\\
&<&\infty.
\IEEEyessubnumber
\label{eq:inter-6-c}
\end{IEEEeqnarray}

Substituting \eqref{eq:inter-6-a}--\eqref{eq:inter-6-c} into
\eqref{eq:inter-6} and \eqref{eq:cond-entropy-Y=HX+N}, we obtain
\begin{IEEEeqnarray*}{rCl}
h_{\sigma}(\hat{\vect Y}\bigr|\abs{\vect X},\abs{\vect Y})>-\infty.
\end{IEEEeqnarray*}

\end{proof}

The main ingredient in the proof of the Lemma~\ref{lem:cap-Y=HX+N}, as
well as  Theorem~\ref{thm:main},  is the following lemma.

\begin{lemma}
Let $\matr M$ be a random matrix and $\vect x\in\Complex^n$ 
a non-zero deterministic vector. If $\matr M$ satisfies the assumptions \eqref{eq:h(M)>-infty}, then
\begin{IEEEeqnarray*}{rCl}
  h(\matr M \vect x)>-\infty.
\end{IEEEeqnarray*}
\label{lem:h(Mx)}
\end{lemma}

\begin{proof}

Since $\vect x\neq 0$, at least one element of $ \vect x$
is nonzero, say $ x_1\neq 0$.
We switch the order of $\matr M$ and $\vect x$ in the product
$\matr M\vect x$ as follows.  
 Let $\vect{ M}\in\Complex^{n^2}$ denote the 
vectorized version of $\matr M$, where rows are concatenated as a
column vector. Define $\vect V\in\Complex^{n^2}$ as follows: 
\begin{IEEEeqnarray}{rCl}
V_{k}=
\begin{cases}
 M_{in}, & r=0,
\\
 Y_{i+r}=\sum\limits_{l=1}^n M_{(i+r)l}x_l, & r=1,
\\
 M_{(i+1)r}, & r\geq 2,
\end{cases}
\IEEEeqnarraynumspace
\label{eq:V-vec}
\end{IEEEeqnarray}
where $k=in+r$, $0\leq i\leq n$, $0\leq r \leq n-1$. Then $\vect
Y=\matr M \vect  x$ is transformed to \eqref{eq:V-vec}, which in
matrix notation is
\begin{IEEEeqnarray}{rCl}
\vect V=\matd A\vect M,
\label{eq:V=AM}  
\end{IEEEeqnarray}
in which $\matd A_{n^2\times n^2}=\diag(\underbrace{\matd
  X,\cdots,\matd X}_{n~\textnormal{times}})$, where the
deterministic matrix $\matd X_{n\times n}$ is
\begin{IEEEeqnarray*}{rCl}
\matd X=\diag\left(
\begin{matrix}
 x_1 & x_2^n\\
 0 & I_{n-1}
\end{matrix}
\right),\quad x_2^n=( x_2,\cdots, x_n),
\end{IEEEeqnarray*}
in which $0$ is the $(n-1)\times 1$ all-zero matrix. 
From \eqref{eq:V=AM}
\begin{IEEEeqnarray*}{rCl}
h(\vect V|\vect{ x})&=&h(\vect M|\vect x)+\log|\det \matd A|\\
&=&h(\vect M)+n\log\abs{x_1}\\
&=&h(\matr M)+n\log\abs{x_1}.
\end{IEEEeqnarray*}
On the other hand, from \eqref{eq:V-vec}
\begin{IEEEeqnarray*}{rCl}
h(\vect V|\vect{ x})&\overset{(a)}{=}&h\left(\vect Y, \{ M_{ij}\}_{j\geq 2}\bigl|\vect{x}\right)\\
&=&h(\vect Y |\vect{ x} )+h\left(\{ M_{ij}\}_{j\geq 2}|\vect x,
  \vect Y\right)\\
&=&h(\matr M\vect x )+h\left(\{ M_{ij}\}_{j\geq 2}|\vect Y\right).
\end{IEEEeqnarray*}
Step $(a)$ holds because conditions $r=1$ and $r=0,1$ in
\eqref{eq:V-vec} include, respectively, $\vect Y$ and $\{M_{ij}\}_{j\geq 2}$. Combining the last two relations
\begin{IEEEeqnarray*}{rCl}
h(\matr M\vect x)= h(\matr M)+ n\log\abs{  x_1}-h\left(\{M_{ij}\}_{j\geq 2}|\vect Y\right).
\end{IEEEeqnarray*}
If  $\E |M_{ij}|^2 < \infty$, from the MET, the
last term is bounded from below. Since $h(\matr M)>
-\infty$ and $x_1\neq 0$, $h(\matr M\vect x)>-\infty$. 

\end{proof}

\begin{lemma}
The random matrix $\matr M$ \eqref{eq:M-expr}, underlying optical
fiber at high powers, satisfies the assumptions of the
Lemma~\ref{lem:cap-Y=HX+N}.
\label{lem:h(M-fiber)>-infty}
\end{lemma}

\begin{proof}

Applying the triangle inequality to \eqref{eq:M-expr}, $|M_{ij}|\leq \bigl(\matd |\matd R|^m\bigr)_{ij}$, where $|\matd R|$ is
the matrix with entries $|r_{ij}|$ and $m$ is the number of stages.

We check the entropy condition in \eqref{eq:h(M)>-infty}. In what follows, let $\theta_{i}\sim\iid\ \mathcal U(0,2\pi)$. For one
linear and nonlinear steps $m=1$:
\begin{eqnarray*}
  \matr M=
 \begin{pmatrix}
   e^{j\theta_1}r_{11} &  e^{j\theta_2} r_{12}\\
   e^{j\theta_1}r_{21} &  e^{j\theta_2}r_{22}
  \end{pmatrix}.
\end{eqnarray*}
In this case, there are four amplitude dependencies $|M_{ij}|=|r_{ij}|$, $1\leq i,j\leq 2$, and two phase dependencies:
\begin{IEEEeqnarray*}{rCl}
 \angle M_{11}=\angle M_{21}+k\pi,\quad \angle  M_{12}=\angle
 M_{22}+k\pi,\quad k=0,1.
\end{IEEEeqnarray*}
A dependency means that $\matr M$ contains a deterministic
component, \ie, $h(\matr M)>-\infty$.

For $m=2$:
\begin{eqnarray*}
  M_{11}&=&     e^{j(\theta_1+\theta_3)}r_{11}^2+  e^{j(\theta_1+\theta_4)}r_{12}r_{21},   \\
M_{12}&=&     e^{j(\theta_2+\theta_3)}r_{12}\left(r_{11}+e^{j(\theta_4-\theta_3)}r_{22}\right),\\
M_{21} &=&  e^{j(\theta_1+\theta_3)}r_{21}\left(r_{11}+e^{j(\theta_4-\theta_3)}r_{22}\right),\\
M_{22} &=&  e^{j(\theta_2+\theta_3)}r_{21}r_{12}+ e^{j(\theta_2+\theta_4)}r_{22}^2.   
\end{eqnarray*}
In this case too, there is a dependency $|r_{21}M_{12}|=|r_{12}M_{21}|$.

For $m=3$:
\begin{IEEEeqnarray*}{rCl}
  M_{11} &=& e^{j(\theta_1+\theta_3+\theta_5)}r_{11}^3+ e^{j(\theta_1+\theta_4+\theta_5)}r_{11}r_{12}r_{21}
   \\
&&+e^{j(\theta_1+\theta_3+\theta_6)}r_{11}r_{12}r_{21}+
    e^{j(\theta_1+\theta_4+\theta_6)}r_{12}r_{21}r_{22},
 \\   
 M_{12}&=&e^{j\theta_2}r_{12}\Bigl(
    e^{j(\theta_3+\theta_5)}r_{11}^2+
+ e^{j(\theta_4+\theta_6)}r_{22}^2
\\
&&+\boxed{ 
e^{j(\theta_3+\theta_6)}r_{12}r_{21}
+
 e^{j(\theta_4+\theta_5)}r_{11} r_{22}}
 \Bigr),
\\
 M_{21} &=& e^{j\theta_1}r_{21}\Bigl(e^{j(\theta_3+\theta_5)}r_{11}^2+
e^{j(\theta_4+\theta_6)}r_{22}^2
\\ 
&&+
\boxed{
 e^{j(\theta_4+\theta_5)}r_{12}r_{21}
+e^{j(\theta_3+\theta_6)}r_{11}r_{22}
}
    \Bigr),
    \\
M_{22} &=&
 e^{j(\theta_2+\theta_3+\theta_5)}r_{11}r_{12}r_{21}+ e^{j(\theta_2+\theta_4+\theta_5)}r_{12}r_{21}r_{22}
    \\
&&+e^{j(\theta_2+\theta_3+\theta_6)}r_{12}r_{21}r_{22}+
    e^{j(\theta_2+\theta_4+\theta_6)}r_{22}^3.
\end{IEEEeqnarray*}
Comparing the boxed terms, $|r_{21}M_{12}|\neq |r_{12}M_{21}|$. There
are still  8 equations for 6 variables. 

In general, the number of entries of $\matr M$ is $n^2$. As $m> 2n$
steps are taken in distance, sufficient number of random variables $\theta_i$ are
introduced in a matrix with fixed dimension. Since $n$ is fixed and
$m$ is free, we obtain an
under-determined system of polynomial equations for
$x_i=\exp(j\theta_i)$ whose solution space has positive
dimension. Thus an entry of $\matr M$ can not be determined from all other
entries. 

\end{proof}

\begin{remark}
The rate interpolation Lemma~\ref{lemm:decomposition} implies that,
replacing $\Complex^n$ by $\mathcal R^+$ changes the asymptotic capacity by
a finite number $c$. From the upper bound $\const C\leq \log(1+\snr)$ in
\cite{yousefi2015cwit2} and Theorem 2.5 in \cite{moser2004dbb}, we
think that the
asymptotic capacity can be achieved by an input distribution that escapes
to infinity. This implies that $\lambda=0$, so that $c$ is indeed zero. We do not
investigate this rigorously.
\qed
\end{remark}

\begin{remark}[Optimal Input Distribution]
Multivariate Gaussian input distribution is a poor choice for channels with 
multiplicative noise. Indeed, it achieves a rate bounded in power in \eqref{eq:Y=HX+N}.  
Log-normal input PDF for the signal norm achieves the asymptotic capacity of the non-constant loss model. 
\end{remark}

%%%%%%%
% Section VI: Review of the Information Theory of the Optical Fiber
%%%%%%%

\section{Review of the Information Theory of the Optical Fiber}
\label{sec:review}

An information-theoretic analysis of the full model of the optical
fiber does not exist. Even in the special case of the zero-dispersion, spectral
efficiency is unknown. In the full model, we do not know anything about
the capacity in the high power regime, let alone the spectral
efficiency. The state-of-the-art is still lower bounds that
are good in the nearly-linear regime. 
This situation calls for basic research, in order to make progress on
these open problems. 

The present paper builds on earlier work. We acknowledge \cite[Eq. 12]{mecozzi1994llh} for
the equation \eqref{eq:zd},  \cite{mecozzi1994llh,turitsyn2003ico,yousefi2011opc } for the PDF of the
zero-dispersion channel, \cite{yousefi2011opc} for the analysis of the
zero-dispersion model, \cite{yousefi2015cwit2,kramer2015upper} for
noting that Shannon entropy is invariant under the flow of a broad class of
deterministic partial differential equations and for highlighting the
usefulness of the operator splitting (in numerical analysis) in the
analysis of the 
NLS equation. Furthermore, we acknowledge \cite{agrell2015conds} for helpful
insight leading to the rate interpolation Lemma~\ref{lemm:decomposition},
\cite{moser2004dbb,lapidoth2003capacity} for the study of the fading
channels and Section~II of \cite{yousefi2012nft3} for unfolding the origin of the capacity
limitations in fiber --- particularly the finding that signal-signal
interactions are not fundamental limitations in the deterministic model if communication takes place in the right
basis (\ie, the nonlinear Fourier basis), which led us to the study of
the remaining factor in this paper, namely the signal-noise
interactions.

We do not intend to survey the literature in this paper. There
is a good review in \cite[Section~I-A]{ghozlan2015focusing}. The achievable
rates of 1- and multi-solitons is studied, respectively, in
\cite{yousefi2012nft3,meron2012soliton,shevchenko2015,zhang2016isit} and \cite{kaza2012,kaza2016soliton,buelow2016}. 
There is also a myriad of lower bounds that hold good in the low power 
regime; see, \eg,  \cite{mecozzi2012nsl,secondini2013achievable,dar2014new,
terekhov2016physrev,secondini2016limitsv2,turitsyn2016nature}.

The achievable rates of the nonlinear frequency-division multiplexing
for multi-user communication  are presented in \cite{yousefi2016nfdm}
for the Hermitian channel. Fig.~\ref{fig:nfdm} compares the NFDM and WDM rates
\cite[Fig.~6]{yousefi2016nfdm}. The gap
between the WDM and NFDM curve reflects signal-signal interactions. The
gap between the NFDM and AWGN curve reflects signal-noise
interactions.  We conjecture that the NFDM rate is  close to the capacity. At the
power levels shown in Fig.~\ref{fig:nfdm}, 
$\const C_{\textnormal{wdm}}(\const P)=\log\const P+c$ and $\const
C_{\textnormal{nfdm}}=\log\const P+c'$, $c<c'$. Although more
gains are expected at $\const P>-2.4$ dB, the slope of the blue curve
will gradually decrease, converging, in the limit $\const
P\rightarrow\infty$, to the asymptotic form in Theorem~\ref{thm:main}.

It is interesting to compare the extent of the signal-noise interactions
in the time domain \cite{serena2016signalnoise} and in the nonlinear Fourier
domain \cite[Section~IV. A]{tavakkolnia2015sig}.

\begin{figure}
\centering{
\includegraphics[width=0.45\textwidth]{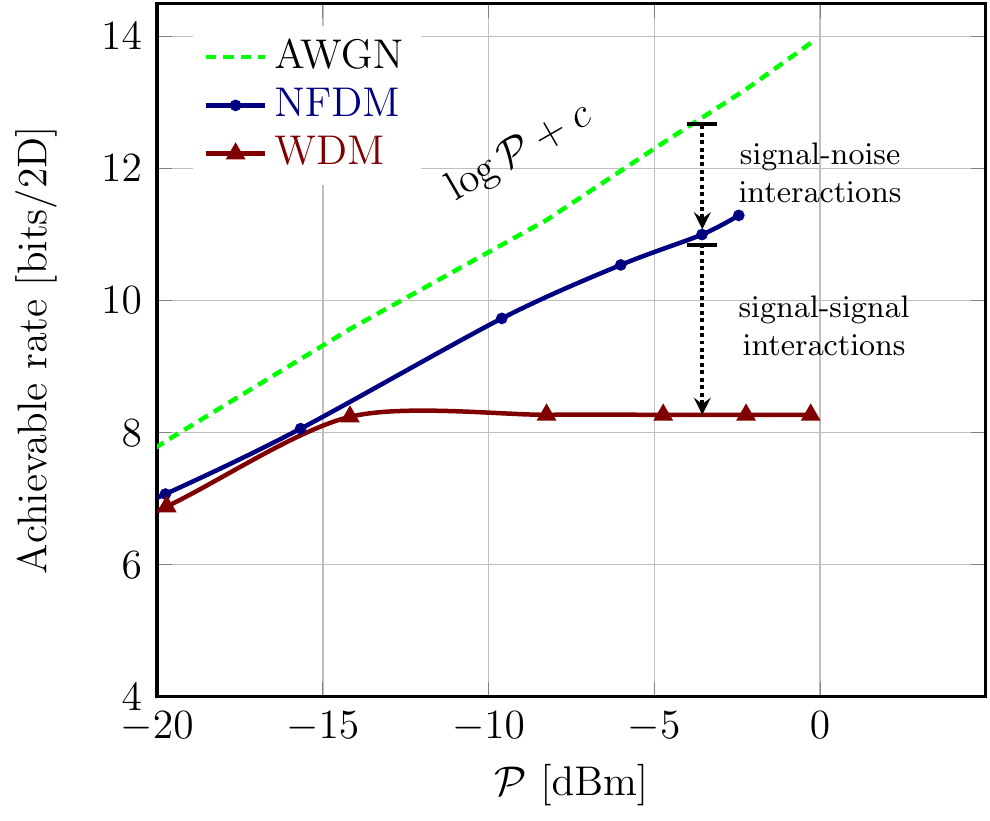}
}
\caption{Achievable rates of the NFDM and WDM, and the capacity of the
  corresponding AWGN channel (from \cite{yousefi2016nfdm}). The NFDM curve is expected to follow
  the asymptotic capacity in the Theorem~\ref{thm:main}.}
\label{fig:nfdm}
\end{figure}

%%%%%%%
% Section VI: Conclusions
%%%%%%%

\section{Conclusions}

The asymptotic capacity of the discrete-time periodic model of the optical
fiber is characterized as a function of the input power in Theorem~\ref{thm:main}.
With $n$ signal \dofs\
at the input, $n-1$ \dofs\ are asymptotically lost, leaving signal energy as the only available \dof\ 
for transmission. The appropriate input distribution is a log-normal PDF for the signal norm.
Signal-noise interactions limit the operation of the optical
communication systems to low-to-medium powers.

%%%%%%%
% Acnowledgments
%%%%%%%

\section*{Acknowledgments}
The research was partially conducted when the author was at the 
Technische Universit\"at  M\"unchen (TUM). The 
support of the TUM Institute for Advanced Study, funded by the German 
Excellence Initiative, and the support of the Alexander 
von Humboldt Foundation, funded by the German
Federal Ministry of Education and Research, are gratefully
acknowledged.  The author thanks Luca Barletta for comments.

\appendices

%%%%%%%
% Appendix A: Proof of the Entropy Identity 
%%%%%%%

\section{Proof of the Identity \eqref{eq:sph-leb}}
\label{app:one}

Let $\der V(\vect x)$ and $\der S(\vect x)$ be the volume and surface element at
point $\vect x\in\Reals^n$ in the spherical coordinate system. Then
\begin{IEEEeqnarray*}{rCl}
\der V(\vect x)&=&|\vect x|^{n-1}\der V(\hat{\vect
  x})\\
&=&|\vect x|^{n-1}\der S(\hat{\vect x})\der|\vect x|.
\end{IEEEeqnarray*}
Thus the Jacobian of the transformation from the Cartesian system with
coordinates $\vect x$ to the spherical system with coordinates $(\abs{\vect x}, \hat{\vect x})$ is
$\abs{\vect x}^{n-1}$. As a consequence
\begin{IEEEeqnarray*}{rCl}
  h(\vect X)&=&h_{\sigma}(\abs{\vect X},\hat{\vect X})+\E\log\abs{\vect
    X}^{n-1}\\
&=&
h(\abs{\vect X})+h_{\sigma}(\hat{\vect X}\bigl|\abs{\vect X})+(n-1)\E\log\abs{\vect
    X}.
\end{IEEEeqnarray*}

%%%%%%%
% Appendix B: MSSFM
%%%%%%%

\section{Input Output Relation in a Unit}
\label{app:in-out-mssfm}
Define
\begin{IEEEeqnarray*}{rCl}
\matr D_1=\diag(e^{j\Psi_k}),\quad \matr D_2=\diag(e^{j\Phi_k}).
\end{IEEEeqnarray*}
The nonlinear steps in Fig.~\ref{fig:mssfm} in matrix notation are 
\begin{IEEEeqnarray*}{rCl}
  \vect U=\matr D_1\bigl(\vect X+\matr{N}^1\vect e\bigr),\quad
 \vect Y=\matr D_2\bigl(\vect V+\matr{N}^2\vect e\bigr),
\end{IEEEeqnarray*} 
where $\vect e\in\Reals^L$ is the all-one column vector. Combining the
linear and nonlinear steps, we obtain \eqref{eq:one-seg} with $\matr M=\matr D_2\matd R\matr D_1$ and
\begin{IEEEeqnarray}{rCl}
\vect Z =\matr M\matr{N}^1\vect e+\matr D_2 \matr{N}^2\vect e.
\label{eq:additive-Z}
\end{IEEEeqnarray}

Clearly $\matr N^{1,2}\vect e\sim\normalc{0}{\D I_n}$. However $\matr
M\matr{N}^1\vect e$ and $\matr D_2 \matr{N}^2\vect e$ are generally non-Gaussian due to the signal and noise terms
in $\Phi_k$ and $\Psi_l$. But, 1)  in the constant loss
model, if 2) $\forall k$ $\vect
x_k\rightarrow\infty$, then
\begin{IEEEeqnarray}{rCl}
\vect Z\sim\normalc{0}{\matd K},\quad \matd K=\D(1+e^{-\alpha_{r}\epsilon})I_n.
\label{eq:noise}  
\end{IEEEeqnarray}
In summary, $N$ variables are Gaussian; $Z$ variables are Gaussians in the
asymptotic analysis of the constant loss model.

\bibliographystyle{IEEEtran}
%\bibliography{Refs}

% Generated by IEEEtran.bst, version: 1.13 (2008/09/30)

\end{document}